\documentclass[sigconf]{acmart}
\usepackage{xr} 
\usepackage{mathtools}
\usepackage{subcaption}
\usepackage{enumitem}
\usepackage{multirow}
\usepackage{graphics}
\usepackage{booktabs}
\usepackage{siunitx}
\usepackage{afterpage}
\usepackage[english]{babel}
\newtheorem{theorem}{Theorem}

\newtheorem{myrule}{Rule}

\usepackage[linesnumbered,ruled]{algorithm2e}
\usepackage[noend]{algpseudocode}
\SetKwFor{Iter}{iter}{do}{end}
\SetKwFor{Concurrent}{concurrently foreach}{do}{end}
\SetKwProg{Pn}{Function}{}{end}
\SetKwProg{Hook}{Event}{}{end}
\algnewcommand\KwAnd{\textbf{\upshape and} }
\algnewcommand{\LeftComment}[1]{\(\triangleright\) #1}
\usepackage{soul}
\usepackage{xcolor}
\newcommand{\stitle}[1]{\noindent \textbf{#1} }

\usepackage{xcolor}  %
\definecolor{Plum}{HTML}{C2938D}
\definecolor{light-gray}{gray}{0.8}
\newcommand{\deleted}[1]{}
\newcommand{\added}[1]{#1}

\usepackage{setspace}
\usepackage{tablefootnote}
\usepackage[colorinlistoftodos,prependcaption,textsize=tiny]{todonotes}
\usepackage{varwidth}

\AtBeginDocument{%
  \providecommand\BibTeX{{%
    \normalfont B\kern-0.5em{\scshape i\kern-0.25em b}\kern-0.8em\TeX}}}

\setcopyright{acmcopyright}
\copyrightyear{2023}
\acmYear{2023}
\acmDOI{XXXXXXX.XXXXXXX}

\acmConference[International Conference on Management of Data '23]{}{June 18--23,
  2023}{Seattle, WA, USA}
\acmPrice{15.00}
\acmISBN{978-1-4503-XXXX-X/18/06}
\author{Ziliang Lai, Chris Liu, Eric Lo}

\affiliation{\institution{The Chinese University of Hong Kong}\country{}}
\email{{zllai,cyliu,ericlo}@cse.cuhk.edu.hk}

\begin{document}
\title{When Private Blockchain Meets Deterministic Database}

\begin{abstract}
Private blockchain as a replicated transactional system shares many commonalities with distributed database. 
However, the intimacy between private blockchain and deterministic database has never been studied. In essence, private blockchain and deterministic database both ensure replica consistency by determinism.
In this paper, we present a comprehensive analysis to uncover the connections between private blockchain and deterministic database.
While private blockchains have started to pursue deterministic transaction executions recently, deterministic databases have already studied deterministic concurrency control protocols for almost a decade.
This motivates us to propose Harmony, a novel deterministic concurrency control protocol designed for blockchain use.
We use Harmony to build a new relational blockchain, namely HarmonyBC, which features low abort rates, hotspot resiliency, and inter-block parallelism, all of which are especially important to disk-oriented blockchain.
Empirical results on Smallbank, YCSB, and \added{TPC-C} show that HarmonyBC offers $2.0\times$ to $3.5\times$ throughput better than the state-of-the-art private blockchains.
\end{abstract}
\fancyhead{}
\maketitle
\setlength{\textfloatsep}{0.0em}
\setlength{\intextsep}{0.5em}
\section{Introduction} \label{sec:intro}

\begin{figure}
    \centering
    \includegraphics[width=0.55\linewidth]{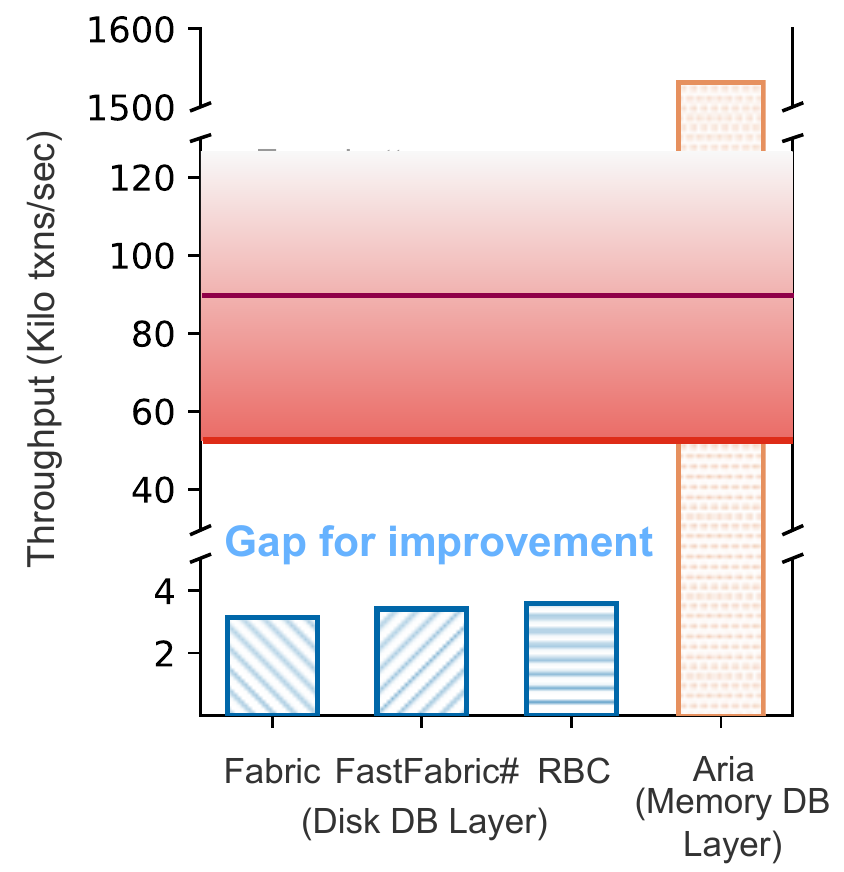}\vspace{-0.5em}
    \caption{\added{The database layer is the bottleneck of disk-based private blockchains. Smallbank workloads on Amazon c5.4xlarge nodes. For WAN, 80 nodes are geo-distributed in four different continents. Throughputs of the database layers are
    measured by using only one ordering node to write off consensus.}}
    \vspace{0.5em}
    \label{fig:bottleneck}
\end{figure}

A blockchain is a ledger \emph{replicated} on every 
peer in the system. 
Internally, it has a \emph{consensus layer} to ensure
all non-faulty replicas agree on the same blocks of the transactions
and process each block of transactions with 
a \emph{database layer}.
The consensus layer
is the bottleneck 
when the database layer is main-memory resident,
i.e., data in DRAM and group commit logs to disk
(e.g., ResilientDB \cite{resilientdb}).
However, the bottleneck stays in the database layer 
when the blockchain is disk-oriented \cite{fabric++}.
Today, almost all enterprise-grade blockchains are disk-oriented (e.g., IBM Fabric \cite{fabric}\deleted{, Fabric++ \cite{fabric++}, Fabric \#\cite{fabricsharp}}, Facebook Diem \cite{diem}, IBM RBC \cite{rbc}, ChainifyDB \cite{chainifydb})
because of a lower deployment cost and 
various use cases
(e.g., IoT device has limited memory \cite{iot}, wimpy systems \cite{wimpy}).
For disk-oriented blockchains, 
while prior work \cite{fabric++,fabricsharp,rbc} has shown the database layer is the bottleneck 
(via overall throughput improvement when optimizing the database layer), 
we confirm that is still the case 
under a more recent setting 
with modern hardware, more nodes, and recent software version
(Figure \ref{fig:bottleneck}).
Specifically, we show that recent (Byzantine) consensus protocols (e.g., HotStuff \cite{hotstuff}, GeoBFT \cite{resilientdb}, RCC \cite{RCC}) have already outrun the \deleted{blockchain database layers}disk-oriented database layers by an order of magnitude.

In this paper, we focus on disk-oriented blockchains
(or simply \emph{blockchains} hereafter)
because of cost and adoption
(e.g., Fabric is available on all major cloud providers \cite{fabric-report} and has over 2000 companies deployed worldwide \cite{fabric-deploy}).
We aim to \deleted{improve its database layer to pursue its consensus layer}\added{narrow the throughput gap between its database layer and the consensus layer}.
The consensus layer often collects transactions from clients 
and broadcasts transaction blocks to all replicas \cite{fabric, rbc, chainifydb}.
Hence, as long as every replica receives and executes the same set of transactions in the database layer \emph{deterministically}, all replicas will reach the same state,
effectively achieving replica consistency.
While concurrent transaction execution is the key to boosting the throughput of the database layer,
classic concurrency control protocols
are insufficient 
to uphold replica consistency due to non-determinism
--- given the same transaction block, 
replicas using the same concurrency control protocol
may reach divergent states
since they may get different serializable schedules.
Consequently, we observe that
state-of-the-art private blockchains 
are actually using different forms of \emph{deterministic concurrency control} implicitly (perhaps subconsciously) to uphold replica consistency with concurrent execution.

The connection between private blockchain and 
\emph{deterministic database}
has never been explicit, although 
the connection between the former and \emph{distributed database}
has been studied elsewhere \cite{bc_vs_db}.
Private blockchain and deterministic database 
actually share a lot in common.
First, both of them are distributed in nature.
Second, 
they also process transactions block-by-block  \cite{calvin,bohm,pwv,aria}.
Third, both blockchains and deterministic databases mandate consistent states across replicas
with blockchains regarding replica consistency as a requirement while
deterministic databases achieve replica consistency as a consequence.

The first contribution of this paper is 
to explicitly connect private blockchain
with deterministic database.
By exploiting the connections, 
we are able to fast-forward the development of private blockchain using lessons learned from deterministic databases.
To demonstrate that,
we follow some recent work \cite{rbc, chainifydb} to  ``chainify'' 
(i.e., modifying a relational database to be a private blockchain \cite{chainifydb})
PostgreSQL 
using the state-of-the-art deterministic concurrency control (DCC) protocol Aria \cite{aria}.
Empirical results demonstrate that such a prototype already yields a 1.3$\times$--2.7$\times$ throughput improvement over the state-of-the-art private blockchains 
(Section \ref{sec:evaluation}). 
That is because Aria after all is the latest DCC protocol
and it admits more concurrency than 
the ones used by the state-of-the-art private blockchains.

Although the preliminary results are exciting, 
 we observe that existing deterministic concurrency control protocols are 
 not optimized for a disk-oriented database layer.
Hence, the second contribution of this paper 
is Harmony, a new deterministic concurrency control protocol
designed for blockchains.
Harmony features 
(i) low abort rates,
(ii) hotspot resiliency, and
(iii) inter-block parallelism, 
all of which are crucial to blockchains.
First, compared to the state-of-the-art deterministic databases,
Harmony is specially optimized to achieve low abort rates
because an abort in a blockchain is way more expensive than an abort in a main-memory database --- a transaction in blockchain involves not only expensive disk I/Os but also network round trips and cryptographic operations. 
Second, modern transactional workloads often suffer from hotspots \cite{bamboo}. 
A hotspot is a handful number of database records that are frequently updated. Hotspots could worsen the number of aborts 
of a concurrency control protocol.
Deterministic concurrency control protocols, however, seldom optimize for hotspots (again due to aborts in main-memory database being relatively cheap)
but Harmony's goal of minimizing aborts can conveniently solve the problem.
Third, unlike all deterministic databases,
Harmony supports inter-block parallelism.
Inter-block parallelism allows a block 
to start processing 
before its previous block ends.
It is advantageous in blockchain 
because the complex interplay 
between I/O, network access, and caching 
would often result in high(er) transaction latency variance
within a block.
Hence, unless with inter-block parallelism,
a blockchain would suffer from resource
under-utilization
that limits its throughput \cite{rbc} (e.g., a straggler transaction in a block 
would idle CPU cores as well as block the pipeline).
Overall, Harmony 
as a deterministic optimistic concurrency control protocol 
has overcome the challenge
of reducing aborts with 
minimal 
serializability check overhead
and
supporting inter-block parallelism.

\begin{table*}
\centering
\resizebox{\textwidth}{!}{%
\begin{tabular}{|l|c|c|c|c|c|}
\hline
\multirow{2}{*}{}                 & \multicolumn{3}{c|}{\textbf{Private blockchains}}                                                                                                             & \multicolumn{2}{c|}{\textbf{Deterministic databases}}                                                                                                                                       \\ \cline{2-6} 
                                  & Fabric\cite{fabric}, Fabric++\cite{fabric++}, Fabric\#\cite{fabricsharp}                                                                                         & RBC\cite{rbc}           & ChainifyDB\cite{chainifydb}              & Calvin\cite{calvin}, BOHM\cite{bohm}, PWV\cite{pwv}                                                                                 & Aria\cite{aria}                                                                                   \\ \hline\hline
Architecture                      & Simulate-Order-Validate                                                                                            & \multicolumn{2}{c|}{Order-Execute}       & \multicolumn{2}{c|}{Sequence-Execute}                                                                                                                                                       \\ \hline
\deleted{
Fault tolerance                   & \multicolumn{3}{c|}{Byzantine-fault or crash-fault tolerant}                                                                                                  & \multicolumn{2}{c|}{Crash-fault tolerant}                                                                                                                                                   \\ \hline}
\begin{tabular}[c]{@{}l@{}}
     Deterministic  \\
     concurrency control
\end{tabular} & \multicolumn{2}{c|}{Optimistic}                                                                                                    & Pessimistic              & Pessimistic                                                                                        & Optimistic                                                                             \\ \hline
Storage                           & \multicolumn{3}{c|}{Disk-based}                                                                                                                               & \multicolumn{2}{c|}{Memory-based}                                                                                                                                                           \\ \hline
Recovery                          & \multicolumn{2}{c|}{Physical logging}                                                                                              & Logical logging          & \multicolumn{2}{c|}{Logical logging}                                                                                                                                                        \\ \hline
\end{tabular}%
}
\caption{Private (disk-based) blockchains and Deterministic databases} 
\vspace{-0.5cm}
\label{fig:bc_vs_dcc}
\end{table*}

The last contribution of this paper is HarmonyBC,
a private blockchain that chainifies PostgreSQL
using Harmony.  
HarmonyBC inherits all the 
features from PostgreSQL such that it is one of the few blockchains that 
can support SQL and stored procedures as smart contracts.
Most 
existing deterministic databases (e.g., Calvin \cite{calvin}, BOHM \cite{bohm}, PWV \cite{pwv}) require a static analysis on 
the stored procedures to 
extract their read-write sets for deterministic scheduling.
However, stored procedures often contain branches that predicate on the query results, which
impede static analysis \cite{aria}.
HarmonyBC does not have such a limitation.
HarmonyBC is extensively evaluated using benchmarks commonly used in blockchains including YCSB \cite{ycsb} and SmallBank \cite{ssi}, and a relational benchmark -- TPC-C \cite{tpcc}.
Empirical results show that \deleted{in an 80-node setting, }HarmonyBC offers
2.0$\times$ to 3.5$\times$ throughput better than RBC \cite{rbc} and FastFabric\# \cite{fabricsharp},
and 2.3$\times$ throughput better than AriaBC (a blockchain
implemented using the same framework as HarmonyBC but using Aria) under high contention.

The remainder of this paper is organized as follows. Section \ref{sec:connect} analyzes the connections between private blockchain and deterministic database.
Section \ref{sec:harmony} describes the design of our deterministic concurrency control protocol Harmony. In Section \ref{sec:harmonybc}, we provide details of HarmonyBC.
Section \ref{sec:evaluation} presents the evaluation. Section \ref{sec:related} discusses related works and Section \ref{sec:conclusion} concludes this paper.

\section{Connecting Private Blockchain with Deterministic Database} \label{sec:connect}
\added{Private blockchain can be used as a shared trustworthy ledger among a group of decentralized and untrusted companies while distributed databases cannot.}
Despite that, they share many commonalities as studied in \cite{bc_vs_db}. 
In this paper, we point out an even closer relationship --
private blockchain \emph{is a} (secure) deterministic database in disguise.
In this section, 
we present comparisons and analyze private (disk-based) blockchain and deterministic database.
We focus on work that optimizes the database layer
in terms of architecture, concurrency, storage, and recovery.
Table \ref{fig:bc_vs_dcc} gives a high-level summary. 
Discussions of other related work
(e.g., consensus, sharding, use of SGX instructions) are in Section \ref{sec:related}.

\subsection{Architecture} \label{sec:architecture}
There are two different types of architecture in private blockchains:
Simulate-Order-Validate (SOV) \cite{fabric, fabric++, fabricsharp} and Order-Execute (OE) \cite{quorum, rbc, chainifydb}.
Deterministic databases only have one type:  Sequence-Execute (SE)  \cite{calvin, bohm, pwv, aria}.

\subsubsection{Simulate-Order-Validate (SOV)}
SOV is 
a blockchain architecture 
advocated by Fabric \cite{fabric}.
Fabric's progeny such as
Fabric++ \cite{fabric++} and Fabric\# \cite{fabricsharp} all follow the same architecture.
In SOV, a transaction has a workflow of
``(1) client $\rightarrow$ (2) endorsers $\rightarrow$ (3) client $\rightarrow$ (4) orderer $\rightarrow$ (5) replicas'' to ensure the transaction is consistently executed on all replicas.

First, (1) a client  submits a  transaction $T$ to a subset of
replicas (known as \emph{endorsers}).
(2) On arriving at an endorser, 
$T$ begins its simulation phase 
and ``executes'' against the endorser's local  \emph{latest state
without persisting its writes}.
The purpose of the simulation phase is to collect the read-write set of $T$, 
where the read-set contains keys and version numbers for all records read and the write-set contains the updated keys and their new values.
Since individual endorsers may catch up with the latest states at different speeds,
the read-write sets for the same transaction $T$ may diverge across the replicas.
Hence, SOV requires each endorser to sign the read-write set it produced and \emph{sends back to the client}.
After collecting the potentially diverged read-write sets of a transaction from the endorsers,
(3) the client follows a predefined policy to pick one read-write set 
and sends it to an ordering service.
(4) The ordering service is often an independent service that   
serves the consensus layer 
to collect transactions from all clients and form 
an agreed block of transactions.
After that, the ordering service broadcasts transactions to all replicas in a 
block-wise manner.
(5) On receiving a transaction block,
a replica validates the security (e.g., verifying the signatures)
and the serializability of the transactions in that block.
For the latter, a replica checks if the received read-set of a transaction is 
still consistent with the replica's local latest version (because it has been a while since the time of getting the read-set in Step (1)), aborts $T$ if it has stale reads
or commits $T$ to update the replica's local state otherwise.

\subsubsection{Order-Execute (OE)}
OE is a blockchain architecture used in
many recent private blockchains like RBC \cite{rbc} and ChainifyDB \cite{chainifydb}. 
In OE, clients submit transactions to an ordering service straight.
The job of the ordering service in OE is the same as in SOV --- 
collects client transactions, 
orders them to form transaction blocks,
and broadcasts blocks to the replicas. 
Unlike SOV, OE only ships \emph{transaction commands} around instead of shipping the transaction read-write sets, which saves network bandwidth.
On receiving a block of transactions, 
each replica executes the block independently.
To uphold consistency across replicas, 
one way is to enforce 
the individual replicas to honor the transaction order in the block by executing the transactions \emph{serially} \cite{quorum}.
ChainifyDB and RBC, however, can execute transactions concurrently (Section \ref{sec:concurrency}).

\subsubsection{Sequence-Execute (SE)}
SE is the architecture of all deterministic databases.
In SE, there is a sequencing layer (which could be as lightweight as a single machine) collecting transactions from client applications and assigning a unique transaction ID (TID) to each transaction. 
Then, the sequencing layer broadcasts the transactions
to the replicas in a block-wise manner.
When a block of transactions reaches a replica, 
every replica executes the transactions using 
a \emph{deterministic concurrency control} protocol, 
with the objective to reach the identical resulting states 
independently without any coordination.

\subsubsection{OE = SE}
At this point, we can see that 
the order-execute architecture from blockchains
is no different from the sequence-execute architecture from deterministic databases:
in terms of functionality, the ordering service in OE is equivalent to 
the sequencing layer of deterministic databases,
whereas the serial/concurrent transaction executions in private blockchains
are actually different types of \emph{deterministic concurrency control} in disguise.

\deleted{
\subsection{Fault tolerance} \label{sec:fault_tolerance}
{\color{red} Early private blockchains use Byzantine-fault tolerant consensus (e.g., PBFT \cite{pbft})
in their ordering service \cite{tendermint,quorum}.
Recent private blockchains assume the ordering service is party-neutral and hence 
simply use a crash-fault tolerance service like Kafka \cite{kafka-performance}.  
Such assumption is acceptable because 
peer identities in private blockchains are known,
and any misbehavior could be identified.}

Deterministic databases are crash-fault tolerant \emph{at best},
with its sequencing layer running 
a crash-fault consensus protocol (e.g. Paxos \cite{paxos})
using multiple nodes.
Since deterministic database and OE-based private blockchain share the same architecture,
making the former Byzantine-fault tolerant is straightforward --- 
replacing its sequencing layer with a Byzantine-fault tolerant ordering service.
{\color{red} Nonetheless, such a change is unnecessary as recent private blockchains 
do not use Byzantine-fault tolerant service either.}
}

\subsection{Deterministic Concurrency Control (DCC)} \label{sec:concurrency}

\begin{table}[]
\centering
\resizebox{\columnwidth}{!}{%
\begin{tabular}{|l|l|}
\hline
\multicolumn{1}{|c|}{\textbf{Pessimistic DCC}}                                                                                                                           & \multicolumn{1}{c|}{\textbf{Optimistic DCC}}                                                                                                                                                                      \\ \hline\hline
\begin{tabular}[c]{@{}l@{}}(a) Deterministic pre-scheduling\\ {\small \quad - Lock manager (Calvin) }\\ {\small \quad - Multi-versioning (BOHM) } \\ {\small \quad - Dependency graph } \\ {\small \quad \; (PWV, ChainifyDB)}\end{tabular} & \begin{tabular}[c]{@{}l@{}}(c) Deterministic read-write sets\\ {\small \quad - Extra-trip back to client} \\ {\small \quad \; (Fabric, Fabric++, Fabric\#)} \\ {\small \quad - Snapshot-based (RBC, Aria)}\end{tabular}                                       \\ \hline
\begin{tabular}[c]{@{}l@{}}(b) Deterministic schedule execution\\       (Calvin, BOHM, PWV, ChainifyDB)\end{tabular}                                                     & \begin{tabular}[c]{@{}l@{}}(d) Deterministic commit \\ \quad - Serial\\ {\small \quad\quad - Dangerous structure (Fabric, RBC)}\\  {\small\quad\quad   - Graph traversal (Fabric++, Fabric\#) } \\ \quad - Parallel\\ {\small \quad\quad   - Dangerous structure (Aria) } \end{tabular} \\ \hline
\end{tabular}%
}
\caption{A taxonomy of DCC protocols}
\label{fig:concurrency}
\end{table}
Concurrency can improve system throughput if properly controlled.
Deterministic databases have been focusing on \emph{deterministic concurrency control} (DCC)
to improve throughput, while the state-of-the-art private blockchains also have similar mechanisms (but not named DCC explicitly).
By observing various sources of non-determinism in private blockchains and deterministic databases, 
we come up with a taxonomy unifying DCC protocols from both deterministic databases and private blockchains.
Table \ref{fig:concurrency} shows the taxonomy, 
in which a DCC protocol is either pessimistic or optimistic.
Pessimistic DCC protocols carefully pre-define a concurrent schedule $S$ for 
a block of transactions prior to their execution
such that every replica can follow $S$ to execute independently.
Pre-defining a concurrent schedule 
for a block of transactions 
requires knowing all their potential conflicts a priori. 
Applying static analysis on a transaction block 
can achieve that to a certain extent.
However, static analysis is insufficient when facing 
workloads with complex transaction logic (e.g., a stored procedure branches based on run-time query results),
which is common in smart contracts.
Optimistic DCC protocols provide deterministic concurrency at run-time,
hence, it requires no static analysis and is thus more suitable to blockchains.
Nonetheless, serializability conflicts may arise at run-time.
Therefore, optimistic DCC protocols would require 
certain transactions to abort and restart to uphold 
serializability.

\subsubsection{Pessimistic DCC (PDCC)}
The idea of pessimistic DCC is to pre-compute a concurrent schedule $S$
for a block $B$ of transactions.
To achieve that, all pessimistic DCC protocols carry out static analysis on an incoming block $B$ to obtain the read-write sets of its transactions to devise the schedule.
Once the read-write sets of the transactions in the block are obtained,
pessimistic DCC can devise a deterministic serializable schedule 
without any aborted transactions.
In practice, however, the read-write sets of the transactions in the block are hard to get using static analysis on real smart contracts. \\

\noindent 
\stitle{Deterministic pre-scheduling (Table \ref{fig:concurrency}a).}~
Calvin \cite{calvin} is an early deterministic database. 
Based on the read-write set 
from static analysis, 
it devises a concurrent schedule $S$ for the {lock manager} to follow
so that the lock manager can grant locks strictly based on the order imposed by the TIDs.
BOHM \cite{bohm} is another deterministic database.
It ``installs'' the devised schedule $S$ in its multi-versioned storage.
It first creates a placeholder for every item that appears in the write-sets, 
with its version number being the TID of the corresponding transaction. 
A transaction $T_i$ (TID=$i$) then reads only items whose version number is the largest 
among all those smaller than $i$ 
or waits until the desired placeholder is filled with a real value.
PWV \cite{pwv} is yet another deterministic database.
It explicitly constructs a dependency graph based on the block's read-write sets and uses that to devise an explicit 
serializable schedule whose commit order follows the TIDs.
ChainifyDB \cite{chainifydb} is a recent private blockchain,
which discusses how to turn an ordinary relational DBMS into a blockchain.
Under the hood, 
ChainifyDB \cite{chainifydb} is similar to PWV,
which explicitly constructs dependency graphs and deterministic
serializable schedules.\\

\noindent 
\stitle{Deterministic schedule execution (Table \ref{fig:concurrency}b).}~
Given a schedule $S$, 
as long as all replicas strictly honor $S$ to execute,
replica consistency can be achieved in a straightforward manner.

\subsubsection{Optimistic DCC (ODCC)} 
ODCC protocols provide deterministic concurrency at run-time and require no static analysis.
Given a transaction $T$, 
ODCC first obtains its \emph{deterministic
read-write set},
despite that the states among the replicas  may diverge due to different reasons (e.g., message delay).
Since the deterministic read-write sets of transactions in a block
may contain serializability conflicts,
ODCC then follows a \emph{deterministic commit} protocol to ensure serializability 
and replica consistency.\\

\stitle{Deterministic read-write sets (Table \ref{fig:concurrency}c).}
SOV and OE (since OE=SE, from now on, we simply focus on OE) 
have different sources of non-determinism when obtaining the read-write sets. 

For SOV-based private blockchains,
since a transaction $T$ is simulated
on a replica using that replica's
local \emph{latest} state,
$T$ may get divergent read-write sets from 
different replicas because different replicas may catch up with the latest states at different rates.
SOV pays the cost of an additional trip back to the client to 
reconcile the divergent read-write sets 
back to a deterministic one.

OE-based blockchains
like RBC \cite{rbc} and 
SE-based deterministic databases 
like Aria \cite{aria} 
use \emph{block snapshots} to ensure all replicas 
get identical read-write sets.
Specifically, 
they leverage the nature of block-based processing 
that the state after processing each block is identical  across replicas.
Hence, instead of obtaining the read-write set of a transaction based on the {latest} state of an individual replica,
they 
obtain the read-write set of a transaction based on
the state after executing a particular block.
We call that state a \emph{block snapshot}.
A block snapshot is deterministic. 
Hence, individual replicas can always use a particular block snapshot as a 
single source of truth to execute transactions independently.
\deleted{
For example, in RBC, a client who wants to submit a transaction $T$ first consults a (random) replica about 
the block number of its last completed block (say the replica replies block 9).
Since then that transaction $T$ would regard the snapshot of block 9 (i.e., 
the state after block 9) as the block snapshot 
to execute transactions on.
Hence, despite $T$ may get broadcasted to a replica late (e.g., the replica 
has already received transactions from block 11)
or early (e.g., the replica 
is still processing block 7),
$T$ would always be get executed against the same snapshot from block 9.
For the latter case, $T$ would wait until the required snapshot is ready.\\
}
\added{
For example, in Aria, a transaction $T$ in block $10$ would regard the snapshot of, say, block 9, as the block snapshot to execute on. Hence, despite $T$ may get broadcasted to a replica $R$ late (e.g., $R$
has already received transactions from block 12)
or early (e.g., $R$
is still processing block 7),
$T$ would always get executed against the same snapshot from block 9.
For the latter case, $T$ would wait until the required snapshot is ready.\\}

\stitle{Deterministic commit  (Table \ref{fig:concurrency}d).} 
Given the identical read-write sets across all replicas,
currently, 
all private blockchains conduct deterministic commit \emph{serially}.
Deterministic databases are way more advanced in this aspect.
They are able to achieve \emph{parallel commit}.

To ensure the commit process is deterministic,
i.e., replicas independently can validate and abort/commit identical sets of transactions,
all blockchains that adopt ODCC 
commit each transaction one by one, serially.
In this regard, Fabric validates transactions in a block in the order of their TIDs and aborts a transaction on seeing a stale read. A stale read can be regarded as a ``dangerous structure'' \cite{ssi} of read(r)-write(w)-dependency (i.e., a transaction reads a record that has been overwritten by another transaction).
However, such a dangerous structure is often overly conservative 
and leads to many false aborts.
Specifically, consider two transactions $T_1$ and $T_2$,
where $T_2$ first reads $x$, followed by $T_1$ updates $x$ but $T_1$ gets committed first.
While the read of $T_2$ is a stale read (hence Fabric would abort $T_2$), $T_2$ indeed can still commit by viewing the resulting serializable schedule as $T_2 \rightarrow T_1$.
Hence, 
RBC leverages that opportunity and validates transactions based on another dangerous structure derived 
from serializable snapshot isolation (SSI) \cite{ssi}.
RBC gets fewer false aborts.
But it still needs to validate transactions serially to uphold determinism.
Validation based on dangerous structures may lead to many false aborts 
because a dangerous structure can only capture local dependencies around a particular transaction, missing the big picture.
Formally, serializability shall actually look for any cycle in the complete transaction dependency graph.
Hence, private blockchains that enhance Fabric (e.g., Fabric++ \cite{fabric++} and Fabric\# \cite{fabricsharp})
construct a dependency graph for the transactions in the ordering service and carry out early validation there.
With a full picture of the dependency graph, there would be fewer aborts, and aborted transactions would not be included in a block and broadcast to the replicas\footnote{Hence, in Fabric++ and Fabric\#, their validation phase only validates the security elements (e.g., signatures) but not serializability.}.
Nonetheless, the expensive and unparallelizable   graph traversal would become a bottleneck when under high contention (Section \ref{sec:evaluation}).

\begin{figure}
    \centering
    \includegraphics[width=0.27\linewidth]{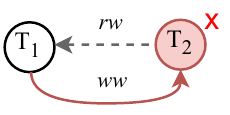}
    \caption{\added{Aria aborts a transaction on seeing a ww-dependency}}
    \vspace{0.5em}
    \label{fig:aria_example}
\end{figure}
Aria as the ODCC inside a deterministic database 
can commit transactions in parallel at the cost of more aborts.
\deleted{
For example, on seeing a write(w)-write(w)-dependency (i.e., two transactions having overlapping write-sets),
Aria aborts the one with a larger TID.
Transactions can check this dangerous structure in parallel by carrying out set-intersections.
After validation, transactions that can be committed always have non-overlapping write sets and thus they can be committed in parallel.
Although Aria can commit transactions in parallel, 
it has a high abort rate in return
--- when there are multiple transactions updating the same record,
only one of them can commit.
In fact, two transactions $T_1$ and $T_2$ that update the same item $x$
can still form a serializable schedule.
}\added{
Specifically, one dangerous structure that Aria is based on is write(w)-write(w)-dependency (i.e., two transactions having overlapping write-sets). 
On seeing a ww-dependency, 
Aria aborts the one with a larger TID.
For example, in Figure \ref{fig:aria_example}, 
Aria aborts $T_2$ on seeing $T_1 \xrightarrow{ww} T_2$.
Breaking ww-dependencies not only avoids cycles (hence serializable) but also allows parallel commit --- after validation, transactions that can be committed always have non-overlapping write sets and thus they can commit in parallel.
However, 
it has a high abort rate in return
--- whenever multiple transactions are updating the same record,
only one of them can commit.
In fact, if the rw-dependency in Figure \ref{fig:aria_example} does not actually exist, Aria still aborts $T_2$ even though there is no cycle.
}

\subsection{Storage} \label{sec:storage}

Main-memory blockchains store data in memory
and group commit logs to disk \cite{resilientdb}.
But enterprise-grade private blockchains are closer to traditional databases 
that store both data and log on disk and use DRAM for caching. 
Disk-oriented blockchains are prevalent in the market
because of cost (Amazon sells DRAM DDR4 ECC for around USD 9.57/GB 
vs. 
NVMe SSD around USD 0.23/GB; 42$\times$ cost gap) and applications.
Deterministic databases are closer to main-memory databases/blockchains
that store data in main memory and use disk for logging.
This explains why Aria as an ODCC protocol is not the best for our setting
because transactions in main-memory databases
are usually short-lived (in the order of $\mu s$) 
and thus weigh parallelism higher than aborts in their protocols \cite{silo}.

\subsection{Recovery} \label{sec:recovery}
There are two classes of logging techniques for recovery: logical logging and physical logging. Logical logging logs only the transaction commands.
Physical logging logs either the read-write sets or ARIES-like redo-undo log.
All SOV-based blockchains use physical logging.
RBC as an OE-based blockchain also uses physical logging.

Deterministic databases use logical logging which has almost no runtime overhead. 
In distributed databases,
parallel recovery from logical logs would generally incur non-determinism, because a replica may get recovered to another serializable database state unless specific treatments are provided \cite{replay}.
But that would not happen in deterministic databases since it can concurrently replay the input transactions with determinism.

\section{Harmony}\label{sec:harmony}
With the understanding of the close connections, 
private blockchains can actually absorb many techniques from deterministic databases. 
In this section, we present Harmony,
a new deterministic concurrency control protocol that supports \emph{parallel commit} with \emph{no application limitation}.
It is superior to the latest DCC Aria when applied to 
disk-oriented blockchain by featuring (i) \emph{low abort rate} (ii) \emph{hotspot resiliency} and (iii) \emph{inter-block parallelism}. 
\added{
Note that improving the concurrency control of a disk-based database at first glance may seem to be a minor issue because 
a database's performance seems to be dominated by its disk latency.
Yet, decades of concurrency control research has shown the otherwise because disk-based databases 
would use all sorts of techniques 
(e.g., DRAM buffer pools and group commit) to hide I/O latency.
With those, higher concurrency (and thus higher CPU utilization) can be translated as better overlapping and I/O hiding,
and reducing aborts can be translated as maximizing the useful work done per I/O, all of which could lead to higher throughput \cite{gray_notes}. These factors motivate the design of better (deterministic) concurrency control for the database layer of disk-based blockchains.
}

\subsection{Overview}
Harmony is an optimistic deterministic concurrency control (ODCC) protocol.
On receiving a block of transactions, it executes transactions in two steps: a simulation step and a commit step, without using static analysis.
Like most ODCCs, its simulation step obtains \emph{deterministic read-write sets} (Table \ref{fig:concurrency}c) by simulating transactions against the same block snapshot. 
After all transactions of the current block finish simulation, it enters the commit step to carry out \emph{deterministic commit} (Table \ref{fig:concurrency}d).

To reduce the abort rate, recent private blockchains have been devoting efforts to exploring broader dependency information (e.g., Fabric++ \cite{fabric++} and Fabric\# \cite{fabricsharp} traverse the whole dependency graph).
However, that is unparallelizable with high overhead especially when the graph is large (Section \ref{sec:evaluation}).
In contrast, Harmony upholds parallelism by only requiring each transaction to examine {dangerous structures} without cross-thread coordination.
To reduce the abort rate, 
Harmony carries out \added{\emph{abort-minimizing validation}} (Section \ref{sec:validation}) that \emph{commits all transactions} as long as their rw-dependencies do not exhibit a kind of new ``\emph{backward dangerous structure}'', while all the other dependencies (e.g., ww-dependencies) are 
managed by \emph{update reordering} (Section \ref{sec:reordering}) to ensure serializability without any abort.
Harmony ensures that the whole process is parallelizable and requires no graph traversal.

Harmony achieves good  resource utilization
by supporting \emph{inter-block parallelism} (Section \ref{sec:inter-block}) --- 
a straggler transaction in block $i$ cannot detain the next block $(i+1)$.
Supporting that requires dealing with a non-deterministic view of inter-block dependency due to network asynchrony. For example, if $T_1$ in block $i$ depends on $T_2$ in block $(i+1)$ (it is possible when these two blocks are concurrent), a replica may miss such inter-block dependency if block $(i+1)$ is delayed, causing it sees different dependencies from the other replicas. Harmony designs an \emph{inter-block abort policy} to ensure deterministic commit under network asynchrony.

\subsection{\added{Abort-minimizing} Validation} \label{sec:validation}
For ease of presentation,
we assume inter-block parallelism is disabled in this section.
In this case, a transaction $T$ starts execution only after the previous block is finished and $T$'s simulation step reads the block snapshot of the previous block. 
We discuss the enabling of inter-block parallelism in Section \ref{sec:inter-block}.

Formally, there are three types of dependencies:
\begin{itemize}

\item \emph{rw-dependency}: 
transaction $T_i$ \emph{rw-depends} on $T_j$ 
if  $T_i$ reads any before-image of $T_j$'s writes,
denoted as $T_i\xrightarrow{rw}T_j$
\added{
\emph{or equivalently}  $T_j \xleftarrow{rw} T_i$. 
The latter format is used when we discuss the ``backward'' dangerous structure later.}

\item \emph{ww-dependency}: 
transaction $T_i$ \emph{ww-depends} on $T_j$
if $T_j$ overwrites any $T_i$'s write, denoted as $T_i \xrightarrow{ww} T_j$.

\item \emph{wr-dependency}:
transaction $T_i$ \emph{wr-depends} on $T_j$
if any of $T_i$'s writes is read by  $T_j$,
denoted as $T_i \xrightarrow{wr} T_j$.
\end{itemize}

We define \emph{rw-subgraph} as the subgraph induced by rw-dependency edges in the dependency graph (i.e., the dependency graph with only rw-dependencies).
Recall that to uphold serializability, no cycles can occur in the whole dependency graph. 
Harmony's abort-minimizing validation focuses on the rw-subgraph only
because other dependencies would be handled by update reordering (Section \ref{sec:reordering}) without any abort.
It is based on the following rule:

\begin{myrule} \label{rule:rw-subgraph}
(Validation Rule) A transaction $T_j$ is aborted if it resides in a \ul{backward dangerous structure}: $T_i \xleftarrow{rw} T_j \xleftarrow{rw} T_k$, $i < j$ and $i \le k$.
\end{myrule}

\begin{figure}
    \quad
    \begin{subfigure}{0.4\linewidth}
    \centering
    \includegraphics[width=0.5\linewidth]{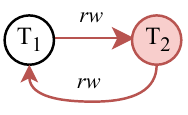}
    \caption{A match with only two transactions}\vspace{-0.1cm}
    \label{fig:match_1}
    \end{subfigure}
    \hfill
    \begin{subfigure}{0.5\linewidth}
    \centering
    \includegraphics[width=0.8\linewidth]{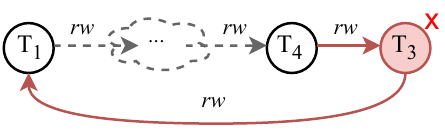}
    \caption{A match with at least three transactions}\vspace{-0.1cm}
    \label{fig:match_2}
    \end{subfigure}
    \quad
    \caption{Examples of backward dangerous structure}
    \vspace{0.5em}
    \label{fig:backward_example}
\end{figure}

Rule \ref{rule:rw-subgraph} is inspired by SSI (serializable snapshot isolation) \cite{ssi, ssi-2, ssi-3}, but with two notable differences:
(1) it imposes a reverse order on the TIDs
\added{(where the oldest transaction $T_k$ with the largest TID comes first 
and the youngest transaction $T_i$ with the smallest TID comes last);}
 and more importantly 
(2) it imposes no constraint on the ww-dependencies (i.e., non-first committer can still win).
Figure \ref{fig:backward_example} shows two examples that match the backward dangerous structure. Notice that a backward dangerous structure can be as small as having two transactions only since we allow $i=k$ (Figure \ref{fig:match_1}). A dangerous structure can also involve an arbitrary number of transactions (e.g., Figure \ref{fig:match_2}, where $T_1 \xleftarrow{rw} T_3 \xleftarrow{rw} T_4$).
Specifically, given any cycle in the rw-subgraph, denote the transaction with the smallest TID in the cycle as $T_i$.
Since $T_i$ is in a cycle, there exists $T_j$ such that $T_i \xleftarrow{rw}T_j$, and there also exists $T_j \xleftarrow{rw} T_k$. Thus, we have $T_i \xleftarrow{rw} T_j \xleftarrow{rw} T_k$, which is a backward dangerous structure because $T_i$ has the smallest TID. Therefore, eliminating all the backward dangerous structures breaks all the cycles in the rw-subgraph.

\deleted{
Formally, we prove the following theorem to demonstrate the correctness of the Validation Rule.

\begin{theorem} \label{thm:rw-subgraph}
The Validation Rule deterministically removes all cycles in the rw-subgraph.
\end{theorem}

\begin{proof}
(1) Determinism: the Validation Rule is deterministic because deterministic read-write sets (obtained based on a block snapshot) induce deterministic rw-dependencies, and the Validation Rule chooses a transaction to abort based on TIDs in the backward dangerous structure.
(2) Cycle removal (by contradiction): assume after applying the Validation Rule on all the transactions, the rw-subgraph still contains a cycle. 
Since a cycle must contain at least two transactions, we discuss two situations.
(a) If the cycle contains only two transactions, 
and the dependencies must be rw-dependencies in the current context,
the cycle must be the one in Figure \ref{fig:match_1}, which is a backward dangerous structure.
(b) If the cycle contains at least three transactions, denote the transaction with the smallest TID in the cycle as $T_i$.
Since $T_i$ is in a cycle, there exist $T_j$ such that $T_i \xleftarrow{rw} T_j$, and there also exists $T_j \xleftarrow{rw} T_k$ because the cycle contains at least three transactions. Thus, we have $T_i \xleftarrow{rw} T_j \xleftarrow{rw} T_k$, which is a backward dangerous structure because $T_i$ has the smallest TID.
Therefore, the cycle would contain a backward dangerous structure regardless of the number of transactions in it, which is a contradiction because it has been eliminated by the Validation Rule.
Hence, Theorem \ref{thm:rw-subgraph} is proved.
\end{proof}
}

\deleted{
Aborting based on the backward dangerous structure is more judicious than the typical dangerous structures found in other ODCCs. Compared to Fabric which aborts on seeing a \emph{single} rw-dependency induced by a stale read, Rule \ref{rule:rw-subgraph} is based on a \emph{pair} of dependencies, which is less likely to match and abort. Compared to RBC and Aria, Rule \ref{rule:rw-subgraph} does not abort any transaction on seeing ww-dependencies.
More importantly,
Rule \ref{rule:rw-subgraph} is not only for cycle removal, but also can ease the computation of update reordering by \emph{converting an expensive topological sort to quick-sort}.
\deleted{
Intuitively, eliminating the \ul{backward} dangerous structures actually
rectifies the rw-dependencies to largely ``go \ul{forward}'' (i.e., follow the TID order) such that sorting the TIDs is almost equivalent to a topological order of the rw-subgraph (more details in Section \ref{sec:reordering}).}
Rule \ref{rule:rw-subgraph} could possibly induce some false aborts (e.g., when the dashed lines do not exist in Figure \ref{fig:match_2}) like the other dangerous-structure-based ODCCs. 
However, eliminating all false aborts would require an expensive and unparallelizable graph traversal (e.g., Fabric++ and Fabric\#), which may outweigh the gain or even become a bottleneck when the dependency graph is complex (Section \ref{sec:evaluation}).
In contrast, Rule \ref{rule:rw-subgraph} can be checked in parallel because each transaction only needs to look at its local dependencies.
}

\added{
Similar to Fabric \cite{fabric}, RBC \cite{rbc}, and Aria \cite{aria},
Rule \ref{rule:rw-subgraph} could induce some false aborts (e.g., when the dashed arrows do not exist in Figure \ref{fig:match_2}). 
Eliminating all false aborts would require an expensive and unparallelizable graph traversal (e.g., Fabric++ \cite{fabric++} and FastFabric\# \cite{fabricsharp}), which may outweigh the gain (Section \ref{sec:evaluation}).
Nonetheless, 
Harmony can induce fewer false aborts than all the dangerous-structure-based ODCCs \emph{in all cases}.
Specifically, compared to Fabric who hastily aborts $T_2$ 
on seeing a \emph{single} rw-dependency
$T_1 \xleftarrow{rw} T_2$, 
Rule \ref{rule:rw-subgraph} is more judicious as it only aborts $T_2$ on seeing \emph{both} $T_1 \xleftarrow{rw} T_2$ \emph{and} $T_2 \xleftarrow{rw} T_3$. %
Compared to RBC and Aria which also abort a transaction based on seeing a pair of rw-dependencies, Harmony would not abort on seeing a ww-dependency. For ww-dependencies, 
Harmony reorders the transactions 
to make them all be able to commit,
instead of aborting them (Section \ref{sec:reordering}).}
\begin{algorithm}
\scriptsize
\SetAlgoLined
\textbf{Initialize} 
 \ForEach{transaction $T_j$}{
   $T_j.\texttt{min\_out} \gets j + 1$\\
   $T_j.\texttt{max\_in} \gets \texttt{-inf}$
}

\LeftComment{In the simulation step:}

\Hook{$on\_seeing\_rw\_dependency (T_i \xleftarrow{rw} T_j)$}{
    $T_j.\texttt{min\_out} \gets \min (i, T_j.\texttt{min\_out})$ \\
    $T_i.\texttt{max\_in} \gets \max (j, T_i.\texttt{max\_min})$
}

\LeftComment{In the commit step:}

\Hook{$on\_entering\_commit (T_j)$ }{
    \uIf{$T_j.\texttt{min\_out} < j$ \KwAnd $T_j.\texttt{min\_out} \le T_j.\texttt{max\_in}$ }{
        Abort($T_j$) \\
    }\uElse {
        Apply\_write\_sets($T_j$) \Comment{See Section \ref{sec:reordering} \quad \quad} \\
    }
 }
\caption{Harmony (no inter-block parallelism)}
\label{alg:validation}
\end{algorithm}

Algorithm \ref{alg:validation} shows the implementation of Rule \ref{rule:rw-subgraph}.
It contains two \emph{event handlers} that would be invoked when a specific event happens  (e.g., HarmonyBC is implemented using PostgreSQL's event API).
In the simulation step, $on\_seeing\_rw\_dependency()$ is invoked when a rw-dependency is found (line \#6).
The handler maintains two variables for each transaction:
\begin{itemize}[leftmargin=2em]
    \item \emph{minimal outgoing TID} of $T_j$: $\texttt{min\_out} = \min \{i | T_i \xleftarrow{rw} T_j, i < j\}$. 
If no $T_i \xleftarrow{rw} T_j (i < j)$, define $\texttt{min\_out} = (j+1)$;
    \item \emph{maximum incoming TID} of $T_j$ : $\texttt{max\_in} = \max \{k | T_j \xleftarrow{rw} T_k\}$. If no $T_j \xleftarrow{rw} T_k$, define $\texttt{max\_in} = \texttt{-inf}$.
\end{itemize}
These two variables are used for checking Rule \ref{rule:rw-subgraph} in the commit step.
When a transaction $T_i$ enters the commit step, $on\_entering\_commit()$ is invoked and line \#12 effectively checks Rule \ref{rule:rw-subgraph}.
If the transaction is not aborted, it invokes $Apply\_write\_sets()$ to apply the write-sets which is detailed in Section \ref{sec:reordering}.
\added{Algorithm \ref{alg:validation} is highly parallel because both the simulation step and the commit step can process transactions concurrently, and the event handlers could be triggered in parallel.} We prove Algorithm \ref{alg:validation} effectively checks Rule \ref{rule:rw-subgraph} as follows:

\begin{proof}
We prove by showing that (a) a transaction aborted by Rule \ref{rule:rw-subgraph} is also aborted by Algorithm \ref{alg:validation} and (b) vice versa.
(a) Consider $T_j$ is aborted due to a backward dangerous structure $T_i \xleftarrow{rw} T_j \xleftarrow{rw} T_k, i < j$ and $i \le k$. By the definition of $\texttt{min\_out}$, $\texttt{min\_out} \le i$; and by the definition of $\texttt{max\_in}$, $\texttt{max\_in} \ge k$. Thus, 
Algorithm \ref{alg:validation} would abort 
$T_j$ because it satisfies the condition in line \#12.
(b) Consider $T_j$ is aborted by Algorithm \ref{alg:validation} due to $\texttt{min\_out} < j$ and $\texttt{min\_out} \le \texttt{max\_in}$. There exists $T_{\texttt{min\_out}}$ and $T_{\texttt{max\_in}}$ such that $T_{\texttt{min\_out}} \xleftarrow{rw} T_j \xleftarrow{rw} T_{\texttt{max\_in}}$ is a backward dangerous structure. Thus $T_j$ is also aborted by Rule \ref{rule:rw-subgraph}. 
\end{proof}

Harmony does not have phantoms because a predicate-read will also trigger $on\_seeing\_rw\_dependency()$ if it induces a rw-dependency.
Let $e$  be the number of rw-dependencies of a transaction $T_j$.
Algorithm \ref{alg:validation} only takes 
$O(e)$ time because each rw-dependency of $T_j$ is only examined once.
In contrast, the validation phase in the original 
SSI takes $O(e^2)$ time \cite{ssi} to  
check every pair of rw-dependencies incident to/from  transaction $T_j$.

\subsection{Update Reordering and Coalescence} \label{sec:reordering}
Given an acyclic rw-subgraph, we now discuss how Harmony ensures the whole dependency graph is acyclic by \emph{update reordering}, and how it achieves parallel commit using \emph{update coalescence}.
We still assume inter-block parallelism is disabled here.

\subsubsection{Update Reordering} \hfill

Recall that Harmony's simulation step obtains deterministic read-write sets like other ODCCs. Harmony keeps the update commands (e.g., add(x, 10)) in the write-set instead of the updated values (e.g., x = 20).
The update commands are collected in the simulation step when a transaction starts working on an \texttt{UPDATE} statement. 
For example, when $T$ starts \texttt{UPDATE} bank \texttt{SET} $\texttt{balance} = \texttt{balance} + 10$ \texttt{WHERE} $\texttt{id}$ = `$\texttt{Alice}$', Harmony extracts the update command of $add(\texttt{Alice.balance}, 10)$ from the physical plan and stores it in $T$'s write-set without evaluating its value. 
Harmony ensures serializability by reordering the update commands using an efficient reordering algorithm.
The update commands are then evaluated one after another following that order.

For example, assume $T_1$ and $T_2$ of the same block update $x$ concurrently by $add(x, 10)$ and $mul(x, 3)$, respectively. Existing snapshot-based ODCCs (Table \ref{fig:concurrency}(c)) would evaluate both updates against the snapshot value of $x$ (say $x=10$). Thus, $T_1$ and $T_2$ would store $x=20$ and $x=30$ in their write-sets, respectively. 
However, one of $T_1$ and $T_2$ has to be aborted because neither $x=20$ nor $x=30$ is a serializable state of committing both 
(e.g., Aria aborts $T_2$ due to $T_1 \xrightarrow{ww} T_2$).
In contrast, Harmony collects $add(x, 10)$ and $mul(x, 3)$ in the simulation step and uses the Reordering Rule (see below) to determine their order in the commit step.
A correct order is necessary to ensure serializability because evaluating $T_2$'s update after $T_1$ induces two dependencies: 
(1) $T_1 \xrightarrow{ww} T_2$. Since $T_2$'s update command $mul(x, 3)$ is a read-modify-write operation, there is another dependency (2) $T_1 \xrightarrow{wr} T_2$.
The induced dependencies could violate serializability
if there is already a dependency $T_1 \xleftarrow{rw} T_2$ in the rw-subgraph, because it will 
form a cycle with (1) and/or (2).
In this case, $T_1$ should update after $T_2$ such that $T_1 \xleftarrow{ww/wr} T_2$ is in the same direction as $T_1 \xleftarrow{rw} T_2$, which does not induce a cycle.
Following this order both transactions can be committed:
$T_2$ first updates $x = mul(x, 3) = 30$, and then $T_1$ evaluates $add(x, 10) = 40$.
That is why Harmony does not have to abort a transaction on seeing a ww-dependency like Aria.
As a tradeoff, evaluating the update commands one after another could impede parallelism. We use \emph{update coalescence} to resolve that (Section \ref{sec:coalesence}).

\deleted{
In addition, with SQL supporting read-modify-write in a single \texttt{UPDATE} statement (e.g., \texttt{UPDATE} bank \texttt{SET} $\texttt{balance} = \texttt{balance} + 10$ \texttt{WHERE} $\texttt{id}$ = `$\texttt{Alice}$'), Harmony is resilient to hotspots if the transactions read-modify-write the hotspot using such \texttt{UPDATE} statements because updates on hotspots do not induce conflicts in Harmony.
As a tradeoff, evaluating the update commands one after another could impede parallelism. We propose update coalescence to resolve that (Section \ref{sec:coalesence}).
}

Intuitively, the order of the update commands has to follow the direction of the dependencies in the rw-subgraph, such that the edges of ww and rw-dependencies also follow that direction to avoid cycles. The following theorem formalizes this principle.

\begin{theorem} \label{thm:reorder}
Given an acyclic rw-subgraph, the complete dependency graph is also acyclic if the update commands are reordered based on the \ul{topological order of the acyclic rw-subgraph}.
\end{theorem}

\begin{figure}
    \centering
    \includegraphics[width=0.6\linewidth]{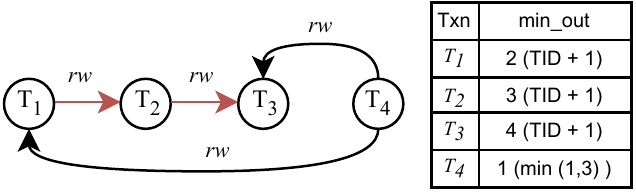}
    \caption{An example of update reordering}\vspace{0.5em}
    \label{fig:reorder}
\end{figure}

\begin{proof}
(By contradiction)
Assume there is still a cycle after reordering the update commands based on the topological order of the rw-subgraph. Then the cycle must contain $T_i \xrightarrow{wr / ww} T_j$ because the rw-subgraph is acyclic.
We only consider $T_i$ and $T_j$ are in the same block because of no inter-block parallelism.
(1) If $T_i \xrightarrow{ww} T_j$, $T_j$'s update is ordered after $T_i$. According to the definition of topological order, there is no directed path from $T_j$ to $T_i$. Thus, $T_i \xrightarrow{ww} T_j$ is not in a cycle, which is a contradiction.
(2) If $T_i \xrightarrow{wr} T_j$, such dependency is only possible when $T_j$ has a read-modify-write update command and that command is ordered after $T_i$'s update.
Therefore, $T_j$ is ordered after $T_i$ in the topological order.
Similar to (1), $T_i \xrightarrow{wr} T_j$ is not in a cycle, which is also a contradiction.
Hence, Theorem \ref{thm:reorder} is proved.
\end{proof}

Although useful, the topological sort in Theorem \ref{thm:reorder} is expensive and hard to parallelize.
Rule \ref{rule:rw-subgraph} now shines by not only ensuring an acyclic rw-subgraph, but also allowing the following Reordering Rule to avoid the topological sort.

\begin{myrule}\label{rule:reorder}
(Reordering Rule)
Given all backward dangerous structures are eliminated 
after Rule \ref{rule:rw-subgraph},
reorder the transactions that update the same record by the ascending order of their minimal outgoing TIDs (i.e., \texttt{min\_out}), and break the tie by their own TIDs.
\end{myrule}

Rule \ref{rule:reorder} transforms an expensive topological sort into a quick-sort (i.e., quick-sort by \texttt{min\_out}), and it enables parallel reordering because updates on different records can be sorted in parallel. 
Figure \ref{fig:reorder} shows an example graph without backward dangerous structures. A topological sort on it results in $[T_4, T_1, T_2, T_3]$.
Suppose only $T_2$ and $T_4$ update $x$, Harmony only needs to quick-sort the \texttt{min\_out}s of $T_2$ and $T_4$ without traversing the whole graph like the topological sort. Nonetheless, the resulting order (i.e., $[T_4, T_2]$) is consistent with the topological ordering.
Besides, if $T_1$, $T_2$, and $T_3$ update $y$, the two updater lists (i.e., $[T_2, T_4]$ on $x$ and $[T_1, T_2, T_3]$ on $y$) can be sorted in parallel. 
It is easy to show that Rule \ref{rule:reorder} is deterministic because the ordering is based on TIDs that are consistent on all replicas.
The following theorem shows the correctness of Rule \ref{rule:reorder}.

\begin{theorem} \label{thm:min_back}
After applying Rule \ref{rule:rw-subgraph}, the ascending order of \texttt{min\_out}s is equivalent to a topological order of the rw-subgraph.
\end{theorem}

\begin{proof}
(By contradiction)
Assume after applying the Validation Rule, $T_j.\texttt{min\_out} < T_i.\texttt{min\_out}$ but $T_j \xleftarrow{rw} T_i$ (i.e., \texttt{min\_out} order is not consistent with the topological order of the rw-subgraph).
(1) If there is no $T_k \xleftarrow{rw} T_j (k < j)$, by the definition of \texttt{min\_out}, $T_j.\texttt{min\_out} = j+1$ and $T_i.\texttt{min\_out} \le j < T_j.\texttt{min\_out}$. Therefore, it contradicts with $T_j.\texttt{min\_out} < T_i.\texttt{min\_out}$.
(2) If there exists $T_k \xleftarrow{rw} T_j (k < j)$, by the definition of \texttt{min\_out}, $T_j.\texttt{min\_out} \le k < j$ and $i > T_i.\texttt{min\_out}$. By the definition of \texttt{max\_in}, $T_j.\texttt{max\_in} \ge i$.
Therefore, $T_j.\texttt{max\_in} \ge i > T_i.\texttt{min\_out} > T_j.\texttt{min\_out}$. Together with $j > T_j.\texttt{min\_out}$, $T_j$ should have been aborted by the Validation Rule. Thus it is a contradiction.
\end{proof}

Reordering is a common technique to reduce aborts.
Rather than reordering some dependencies after the \emph{values} are computed \cite{aria,fabric++,fabricsharp,tictoc,batch-occ},
Harmony can reorder more by evaluating
the \emph{commands} in the commit step.
IC3 \cite{ic3} and ROCOCO \cite{atomic-commit} can achieve optimal reordering (i.e., minimal aborts). 
Yet they first need a static analysis to get the complete picture,
which is not suitable for many blockchain workloads.

\begin{figure}
    \centering
    \includegraphics[width=0.75\linewidth]{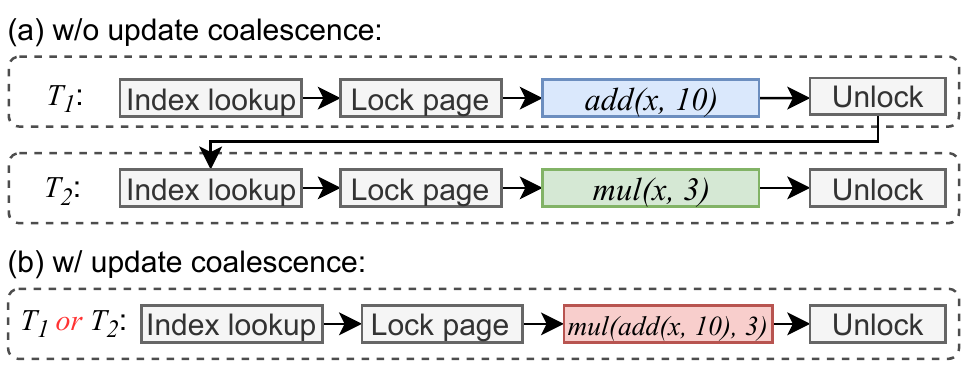}
    \caption{Physical plans of the update commands with and without coalescence}\vspace{0.5em}
    \label{fig:coalesce}
\end{figure}
\subsubsection{Update coalescence} \label{sec:coalesence}\hfill\\
Although Rule \ref{rule:reorder} efficiently computes a serializable update order, 
efficiently enforcing that order when applying the update commands needs to solve two problems.
First, 
transactions updating the same record would have unnecessarily duplicated disk I/Os and locking (problem P1). 
Figure \ref{fig:coalesce} shows an example where $T_1$ and $T_2$ update $x$ by $add(x,10)$ and $mul(x,3)$, respectively, and $T_2$'s update is ordered after $T_1$.
Without update coalescence (Figure \ref{fig:coalesce}a), $T_2$'s physical plan of updating $x$ largely duplicates what $T_1$ does (e.g., index lookup and locking).
Second,
applying update commands one after another impedes parallelism (problem P2). In Figure \ref{fig:coalesce}a, since $T_2$ is ordered after $T_1$, $T_2$ would be detained by $T_1$ especially when $T_1$ experiences disk I/Os or $T_1$ itself is detained by another transaction.

We propose update coalescence that merges multiple update commands on the same record into one to mitigate the above problems.
Concretely, Harmony eliminates duplicated operations and merges the physical plans of updating the same record without affecting the semantics.
As shown in Figure \ref{fig:coalesce}b, the physical plan after coalescing encodes both $T_1$ and $T_2$'s updates on $x$ and respects the order of $T_1$ followed by $T_2$.
In this way, Harmony resolves P1.
To resolve P2, Harmony dedicates only one transaction to apply the coalesced update for each record. 
In Figure \ref{fig:coalesce}b, if $T_1$ works on $x$ first, it applies the coalesced update while $T_2$ can simply skip and work on other updates in parallel.

\begin{algorithm}
\scriptsize
\SetAlgoLined

\LeftComment{In the simulation step:}

\Hook{$on\_update (\texttt{key}, \texttt{update\_command})$}{
    $\texttt{update\_cmds} \gets \texttt{update\_reservation}.search(\texttt{key})$ \\
    $\texttt{update\_cmds}.append(\texttt{update\_command})$\\
    $T_{current}.\texttt{updated\_keys}.append(\texttt{key})$
}

\LeftComment{In the commit step, invoked by Algorithm \ref{alg:validation}}

\Pn{$Apply\_write\_sets (T)$ }{
    \ForEach{$\texttt{key} \in T.\texttt{updated\_keys}$}{
        $\texttt{update\_cmds} \gets \texttt{update\_reservation}.search(\texttt{key})$ \\
        \If{$\texttt{update\_cmds}.\texttt{handled} == \mathbf{False}$}{
            $\texttt{update\_cmds}.\texttt{handled} \gets \mathbf{True}$ \\
            $\texttt{update\_cmds}.filter(\texttt{update\_command}.T.aborted == \mathbf{False})$ \\
            $\texttt{update\_cmds}.sort\_by(\texttt{update\_command}.T.\texttt{min\_out})$\\
            $\texttt{coalesced\_update} \gets coalesce(\texttt{update\_cmds})$ \\
            $apply(\texttt{coalesced\_update})$
        }        
    }
 }
\caption{Update Reordering and Coalescence}
\label{alg:apply-writes}
\end{algorithm}

Algorithm \ref{alg:apply-writes} shows how update reordering and update coalescence fit into Harmony.
In Algorithm \ref{alg:apply-writes}, Harmony uses a
hash table (\texttt{update\_reservation}) to map a key to the list (\texttt{update\_cmds}) of the update commands related to it.
When the event handler $on\_update()$ is triggered (line \#2), 
the update command is appended to the corresponding \texttt{update\_cmds} list (line \#4).
The updated key is also appended to the \texttt{updated\_keys} list of the transaction (line \#5).
The function $Apply\_write\_sets(T)$ does the actual work of update reordering and coalescence, and it is invoked by Algorithm \ref{alg:validation} after $T$ passes the validation.
For each updated key of $T$ (line \#9), the corresponding \texttt{update\_cmds} is checked to see if it has already been handled by another transaction (line \#11).
Lines \#11 and \#12 ensure that for each updated record, only one transaction handles all updates related to it, and other transactions can skip and work on other updates in parallel (lines \#11 and \#12 are protected in a critical section).
To apply the updates, $T$ first filters out the update commands added by the aborted transactions in the \texttt{update\_cmds} (line \#13), and then it performs update reordering by sorting the \texttt{update\_cmds} based on the \texttt{min\_out}s as suggested by Rule \ref{rule:reorder} (line \#14).
After that, it coalesces the update commands (as demonstrated in Figure \ref{fig:coalesce}) and applies the coalesced update (line \#16).

There are two corner cases that are not explicitly handled in Algorithm \ref{alg:apply-writes}: (1) $T$ reads an $x$ that has been updated by itself via a command $op(x)$; and (2) $T$ updates $x$ more than once. For (1), $T$ retrieves $op(x)$ in the \texttt{update\_reservation} table and evaluates it for the read operation.
Notice that $op(x)$ may be evaluated twice in this case (i.e., once for the read in the simulation step and once after update reordering in the commit step), but both evaluations are guaranteed to return the same value since Rule \ref{rule:reorder} respects serializability. For (2), when $T$ inserts the second update command to $x$'s \texttt{update\_cmds} list, it coalesces the second update command with the first one such that the \texttt{update\_cmds} effectively only contains at most one update command for each transaction.

With update reordering and coalescence, Harmony 
achieves hotspot resiliency because 
the former allows all concurrent updaters to commit, and 
the latter coalesces many updates of a hotspot into one.
IC3 \cite{ic3} also merges commands to reduce overhead.
Yet, it is based on static analysis.
Overall, update reordering and coalescence 
handle updates 
at the command level rather than at the value level,
giving them a complete picture to optimize.
Nonetheless, those opportunities might be 
lost in case a smart contract developer 
separates the read-modify-write logic.
For example, consider 
\texttt{UPDATE ... SET X=X+1 WHERE Y=10},
which can be reordered and coalesce with other update statements.
However, the opportunity would be lost if 
the developer expresses the logic 
as three pieces in a stored procedure: 
(1) first read the value of {\tt X} using a SQL {\tt SELECT},
(2) increment {\tt X}, and
(3) finally write the updated value of \texttt{X} using a 
SQL {\tt UPDATE}.
For some cases, the query optimizer may be able to rewrite 
them into one.
Otherwise, smart contract developers are encouraged to 
express the entire read-modify-write logic as one SQL statement.

\deleted{
Since update reordering and coalescence optimize the \texttt{UPDATE} statements, they would work better if the read-modify-write operation is composed into one \texttt{UPDATE} statement (e.g., \texttt{UPDATE} bank \texttt{SET} \texttt{balance} = \texttt{balance} + 10 \texttt{WHERE} $\texttt{id}$ = `$\texttt{Alice}$') instead of separating it into a \texttt{SELECT} and an \texttt{UPDATE}.
}

\subsection{Inter-block Parallelism} \label{sec:inter-block}
Inter-block parallelism \cite{fabric,rbc,schain}
aims to improve resource utilization in case a straggler transaction in block $(i-1)$ detains the next block $i$. 
Concretely, 
when block $(i-1)$ spares some resources (e.g., most transactions in block $(i-1)$ are finished but a straggler is still running), Harmony would start some transactions from block $i$ to utilize the spared resources before the whole block $(i-1)$ finishes. When inter-block parallelism is enabled, transactions in block $i$ carry out the simulation step based on the snapshot of block $(i-2)$ instead of block $(i-1)$, since the latter may not be ready when block $i$ starts.
Although adjacent blocks could run concurrently, Harmony still runs the commit step of block $(i-1)$ before the commit step of block $i$ to uphold determinism.
As a tradeoff, inter-block parallelism may introduce inter-block dependencies that increase aborts. 
But our experiments (Section \ref{sec:effectiveness}) confirm that better resource utilization outweighs that because Harmony is especially good at reducing the abort rate. 
\deleted{
One reason is that Harmony minimizes aborts. Hence, overlapping the adjacent blocks would not significantly increase the number of dependencies. 
}

Harmony's validation (Rule \ref{rule:rw-subgraph}) has to be enhanced to handle inter-block dependencies. For example, if $T$ in block $i$ reads $x$ from the snapshot of block $(i-2)$ and $T'$ in block $(i-1)$ updates $x$, there is an inter-block rw-dependency denoted as $T' \xleftarrow{inter-rw} T$ (we add a prefix of ``intra'' or ``inter'' for clarity). 
Moreover, inter-block dependencies may cause non-determinism due to network asynchrony. For example, Figure \ref{fig:inter-block} shows a backward dangerous structure. A replica R1 who sees the whole backward dangerous structure would \ul{abort} $T_2$.
However, if block $(i+1)$ arrives late in replica R2
due to network asynchrony, R2 would \ul{commit} $T_2$ because when applying 
Rule \ref{rule:rw-subgraph}, 
R2 only sees $T_1 \xleftarrow{intra-rw} T_2$. This causes inconsistency between $R1$ and $R2$. 

\begin{figure}
    \centering
    \includegraphics[width=0.45\linewidth]{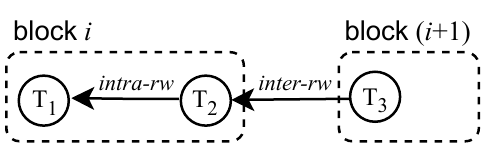}
    \caption{An example that may cause non-determinism if without the Enhanced Validation Rule}\vspace{0.5em}
    \label{fig:inter-block}
\end{figure}

To deal with the inter-block dependencies (denoted as \emph{inter-dep}), 
Harmony uses an enhanced validation rule.

\begin{myrule} \label{rule:enhanced-reorder}
(Enhanced Validation Rule)
For each \ul{generalized backward dangerous structure}: $T_i \xleftarrow{inter-dep / intra-rw} T_j \xleftarrow{inter-dep / intra-rw} T_k$, $i<j$ and $i \le k$, apply the following abort policy:
\begin{enumerate}[label=(\roman*)]
    \item if $T_j$ and $T_k$ are in the same block, abort $T_j$;
    \item othewise, abort $T_k$.
\end{enumerate}
\end{myrule}

Rule \ref{rule:enhanced-reorder} is identical to Rule \ref{rule:rw-subgraph} if there are no inter-block dependencies. 
The additional abort policy (ii) ensures determinism despite network asynchrony. Specifically, consider the case in Figure \ref{fig:inter-block}, if block $(i+1)$ is not delayed in the replica R1, it would abort $T_3$ instead of $T_2$
according to the abort policy (ii); suppose another replica R2 sees $T_1 \xleftarrow{intra-rw} T_2$ first and then sees $T_2 \xleftarrow{inter-rw} T_3$ later due to the delayed block $(i+1)$. R2 would also abort $T_3$. This ensures that R1 is consistent with R2.
We now prove that Rules \ref{rule:enhanced-reorder} and  \ref{rule:reorder} together ensure serializability and they are deterministic.

\begin{proof}
(1) Serializability: 
After applying both rules, 
we show that (a) there is no cycle within each block, and (b) there is no cycle across blocks. The correctness of (a) follows the fact that Rule \ref{rule:enhanced-reorder} is equivalent to Rule \ref{rule:reorder} when no inter-block dependencies are considered. 
For (b), 
Rule \ref{rule:enhanced-reorder} considers all the backward dangerous structures across blocks by adding inter-block dependencies into the generalized backward dangerous structure. Similar to how 
Rule \ref{rule:rw-subgraph} breaks the cycles in the rw-subgraph by eliminating backward dangerous structures, eliminating the generalized backward dangerous structure breaks the cycles across blocks.
(2) Determinism: even with inter-block parallelism, Harmony enforces the blocks to enter their commit steps in the order of block ID. Thus, the blocks apply Rules \ref{rule:enhanced-reorder} and \ref{rule:reorder} in a deterministic order. Since we have shown that Rules \ref{rule:enhanced-reorder} and \ref{rule:reorder}  are deterministic, we conclude that Harmony is deterministic with inter-block parallelism.
\end{proof}

\section{HarmonyBC}\label{sec:harmonybc}
In this section, we describe in detail how we build a private blockchain, namely HarmonyBC, using Harmony DCC.

\stitle{Architecture.}
Recall that private blockchains have two major architectures: Simulate-Order-Validate (SOV) and Order-Execute (OE). 
HarmonyBC adopts the OE architecture because it has no network overhead of sending large read-write sets.
More importantly, Harmony DCC fits the OE architecture better because it only requires blocks of transaction commands as input.

Similar to the state-of-the-art relational private blockchains \cite{rbc, chainifydb}, we implement Harmony DCC on
PostgreSQL \cite{postgresql}, a disk-based relational database as the database layer.
Therefore, HarmonyBC has the full functionalities of a relational database. Since PostgreSQL is process-based, 
concurrency is achieved by executing all transactions in a block using
multiple processes in parallel.

\stitle{Consensus and fault-tolerance.}
The consensus layer of HarmonyBC is a pluggable module and we currently support the Byzantine-fault tolerant HotStuff \cite{hotstuff} and crash-fault tolerant Kafka. 
Private blockchains typically rely on the consensus layer to uphold the safety and liveness guarantees \cite{blockchain-book, fabric, rbc}.
With a BFT consensus layer (e.g., HotStuff \cite{hotstuff}), HarmonyBC is also BFT because a faulty database node can only tamper its own state without affecting the non-faulty majority.

\stitle{Recovery.}
HarmonyBC enjoys lightweight logical logging for crash-fault recovery 
because of determinism. 
It persists the small input blocks before execution instead of the large ARIES-like log. 
It utilizes the \texttt{CHECKPOINT} command in PostgreSQL to flush dirty pages to disk every $p$ blocks (e.g., $p$ = 10).
After checkpointing, the ID of the latest checkpointed block is also persisted into a \texttt{block\_checkpoint\_log}.
During recovery, the replica loads the latest checkpoint and re-executes the blocks after the latest checkpointed block.
When performing checkpointing, the previous checkpoint is not overwritten via PostgreSQL's multi-versioned storage. Therefore, if a replica crashes during checkpointing, it can still recover from the previous checkpoint.

\stitle{Security.}
Private blockchains have to support node authentication to allow only identified nodes to join and must ensure tamper-proof \cite{fabric}.
For node authentication, we reuse the user authentication in the consensus layer (e.g., Kafka) such that only identified clients can submit transactions.
The replicas are also authenticated when connecting to the consensus layer.
For tamper-proof, since the input determines the final states in DCC, ensuring a tamper-proof input guarantees the tamper-proof of the final state.
Therefore, HarmonyBC includes in each block a hash of the previous block like a typical blockchain, such that any tampered block could be identified by back-tracing the hash values from the latest block. 

\stitle{Limitations.}
Harmony inherits some limitations from the OE architecture.
First, its trust model is fixed (i.e, majority) while SOV blockchains are able to program the trust logic in the \emph{endorsement policy} (e.g., it can define one-vote veto).
We plan to support it in future work by allowing users to define the endorsement policy via a stored procedure.
Second, it does not support workload partitioning like SOV blockchains \cite{fabric}.
In future work, we also plan to employ sharding  \cite{sharding,sharding-2,sharding-3}.

\section{Evaluation} \label{sec:evaluation}
We use two types of clusters. (1) Default cluster: 7 machines with Intel Xeon E5-2620v4, 64GB DRAM and 800GB SSD; nodes are connected using 1Gbps Ethernet.
This mimics the default setting of recent disk-based blockchains \cite{rbc,fabric++,fabricsharp}. 
(2) Cloud cluster: 80 \texttt{t3.2xlarge} instances on AWS, each of which has 8 vCPUs, 32GB DRAM, and 30GB SSD. Nodes are either located in LAN (5Gbps Ethernet) or in WAN (across 4 regions on 4 different continents).  We use that cluster to study the scalability.
All nodes run 64-bit CentOS 7.6 with Linux Kernel 3.10.0 and GCC 4.8.5.

We compare HarmonyBC with two SOV private blockchains:
(1) \ul{Fabric} v2.3 and (2) \ul{FastFabric\#} \cite{fabricsharp} (the latest progeny of Fabric; 
an optimized implementation of Fabric\# \cite{fabricsharp}; also better 
than Fabric++ \cite{fabric++} and FastFabric \cite{fastfabric}). We use the code provided by the author.
We follow \cite{tune-fabric} and 
tune the system parameters to optimal.
We also include an OE blockchain (3) \ul{RBC} \cite{rbc}. Since RBC is not open-sourced, we implemented it using the same framework as HarmonyBC.
Besides, we also chainified PostgreSQL using 
 Aria \cite{aria} as an OE private blockchain, namely (4) \ul{AriaBC}.
AriaBC is also implemented using the same framework as HarmonyBC 
because its original implementation is only an \emph{in-memory} standalone implementation without integrating into a real system.
Hence, all OE blockchains in this study are implemented using the same framework. 
All OE-based blockchain implementations (including HarmonyBC) provide the same security guarantees as the SOV-based blockchains and use the same consensus layer (Kafka by default). 
In fact, all SOV-based blockchains are given advantages 
because they use simpler key-value storage while RBC, AriaBC, 
and HarmonyBC are relational.
Systems that require static analysis (Calvin \cite{calvin}, ChainifyDB \cite{chainifydb}, PWV \cite{pwv}, and etc \cite{bohm,pwv,schain,quecc,caracal,single-thread-dcc,tpart})
are not included because of their application limitations.
We exclude blockchains that rely on trusted hardware (e.g., CCF \cite{ccf} uses SGX)
because HarmonyBC does not make any assumption on trusted hardware.

\begin{table}
\centering
\resizebox{0.85\linewidth}{!}{%
\begin{tabular}{c|cccccc} \hline
    $\texttt{skewness}$ & 0 & 0.2 & 0.4 & \textbf{0.6} & 0.8 & 1.0 \\ \hline
    YCSB      & 1.1\%  & 1.2\% & 2.4\% &  9.9\% & 38.3\% & 74.3\%  \\
    Smallbank & 0.1\%  & 0.1\% & 0.2\% &  1.5\% & 2.8\% & 10.6\%   \\ \hline\hline
    $\texttt{warehouses}$ & & 1 & \textbf{20} & 40 & 60 &  80  \\ \hline
    TPC-C      & & 47.9\%  & 8.8\% & 4.2\% & 2.7\% & 2.4\%   \\\hline
\end{tabular}
}
\caption{\added{Hit rate of the backward dangerous structure}}\label{tab:hit_rate}
\end{table}

We use YCSB \cite{ycsb}, Smallbank \cite{smallbank}, and \added{TPC-C \cite{tpcc}} benchmarks.
For YCSB, we set the number of keys to 10K and follow \cite{mocc,tictoc,aria} to wrap 10 operations into one transaction. Each operation has equal probabilities of being a simple \texttt{SELECT} or \texttt{UPDATE}.
For Smallbank, we also set the number of accounts to 10K and use the standard mix. 
Unless stated otherwise, the block size is tuned to be optimal for each setting and the contention is set to medium ($\texttt{skewness} = 0.6$ in YCSB and Smallbank, \added{$\texttt{warehouses}=20$ in TPC-C)}.
\added{All workloads can exercise HarmonyBC's dangerous structure in various degrees (see Table \ref{tab:hit_rate}).}
\added{
Since Fabric and FastFabric\# are not relational
(do not natively support TPC-C \cite{fabric-tpcc}), 
we use a separate section to discuss TPC-C results
(Section \ref{sec:tpcc}), in which Fabric and FastFabric\# are excluded.}

\subsection{Overall performance}

In this experiment, we measure the peak throughput and the end-to-end latency of the committed transactions for all systems using the default cluster.

Figures \ref{fig:smallbank_overall} and \ref{fig:ycsb_overall} show the results.
HarmonyBC attains 3.5$\times$ and 2.0$\times$ throughput over the best of the existing private blockchains (i.e., RBC in this experiment) in Smallbank and YCSB, respectively.
It also achieves around 70\% lower latency than SOV blockchains (Fabric and FastFabric\#) because the OE architecture has fewer round-trips.
RBC has similar latency as HarmonyBC but attains lower throughput because RBC admits a lower level of concurrency (see Section \ref{sec:exp_block}).

Simply implementing Aria as a private blockchain (i.e., AriaBC) readily yields better throughput than existing private blockchains, which shows the benefits of using an advanced DCC from a deterministic database to improve private blockchain. 
HarmonyBC attains a larger margin (e.g., 1.5$\times$ throughput over AriaBC in YCSB), demonstrating the effectiveness of HarmonyBC's blockchain-specific optimizations (e.g., low abort rate and inter-block parallelism).
AriaBC has a slightly larger latency than HarmonyBC because its optimal block size is larger (see section \ref{sec:exp_block})

\begin{figure}
     \centering
     \includegraphics[width=0.95\linewidth]{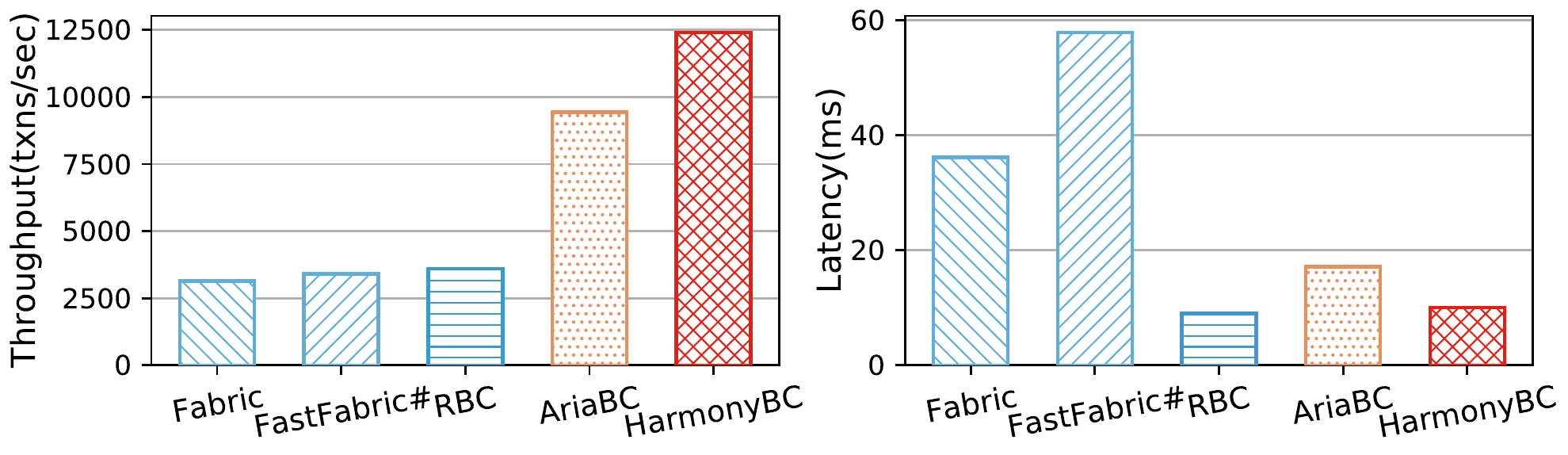}
     \vspace{-0.5em}
     \caption{Overall performance on Smallbank}
     \vspace{-1em}
     \label{fig:smallbank_overall}
\end{figure}

\begin{figure}
     \centering
     \includegraphics[width=0.95\linewidth]{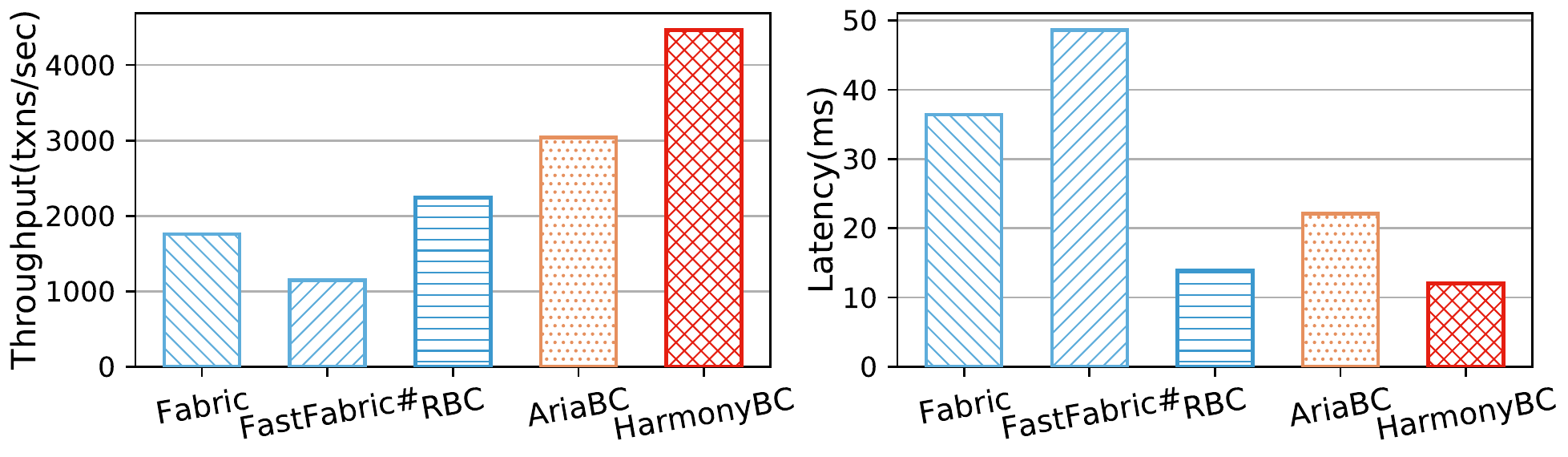}
     \vspace{-0.5em}
     \caption{Overall performance on YCSB}
     \vspace{0.5em}
     \label{fig:ycsb_overall}
\end{figure}

We observe that FastFabric\# outperforms Fabric in Smallbank but the opposite in YCSB. 
The reason is two-fold. First,
the latest version of Fabric used in our experiments has a higher throughput than that used in FastFabric\# paper. 
Second,
a transaction accesses 10 records in YCSB (at most two in Smallbank), resulting in a much more complicated dependency graph.
Our profiling shows the end-to-end throughput is bounded by the graph traversal indeed (a transaction spends around 75\% of runtime on it).
For the same reason, FastFabric\# has higher latency than Fabric.

\subsection{The impact of block size and concurrency}  \label{sec:exp_block}

Figures \ref{fig:smallbank_vary_block_size} and \ref{fig:ycsb_vary_block_size} show the results when varying block sizes (number of transactions per block) in our default cluster.
The block size is also the degree of concurrency 
for blockchains that support concurrent transactions
(i.e., HarmonyBC, AriaBC, and RBC)
because they can execute all transactions of a block in parallel
using one process/thread per transaction.

When the block size is very small (e.g., 5),
the throughput is low because of the limited degree of concurrency.
The throughput also drops when the block size is too large,
because of the increased number of conflicts and lock contentions.
The throughput drops are less noticeable in Smallbank (an almost flat curve in Figure \ref{fig:smallbank_vary_block_size} when block size $>50$) because transactions in Smallbank access fewer records and thus have fewer conflicts.
The latency increases with the block size due to the longer time to form and process a larger block.
This effect is most serious for FastFabric\# because larger block sizes induce larger dependency graphs, causing it to spend more time on graph traversal.

HarmonyBC's optimal block size is 25 for both YCSB and Smallbank, while RBC has a smaller optimal block size (i.e., 10) because it commits transactions serially such that a larger block size does not increase the level of concurrency much. 
In contrast, 
AriaBC's optimal block size is larger (50 for YCSB and 75 for Smallbank) because it favors more concurrency since Aria is originally designed for a highly concurrent main-memory database.
\begin{figure}[t]
     \centering
     \includegraphics[width=0.9\linewidth]{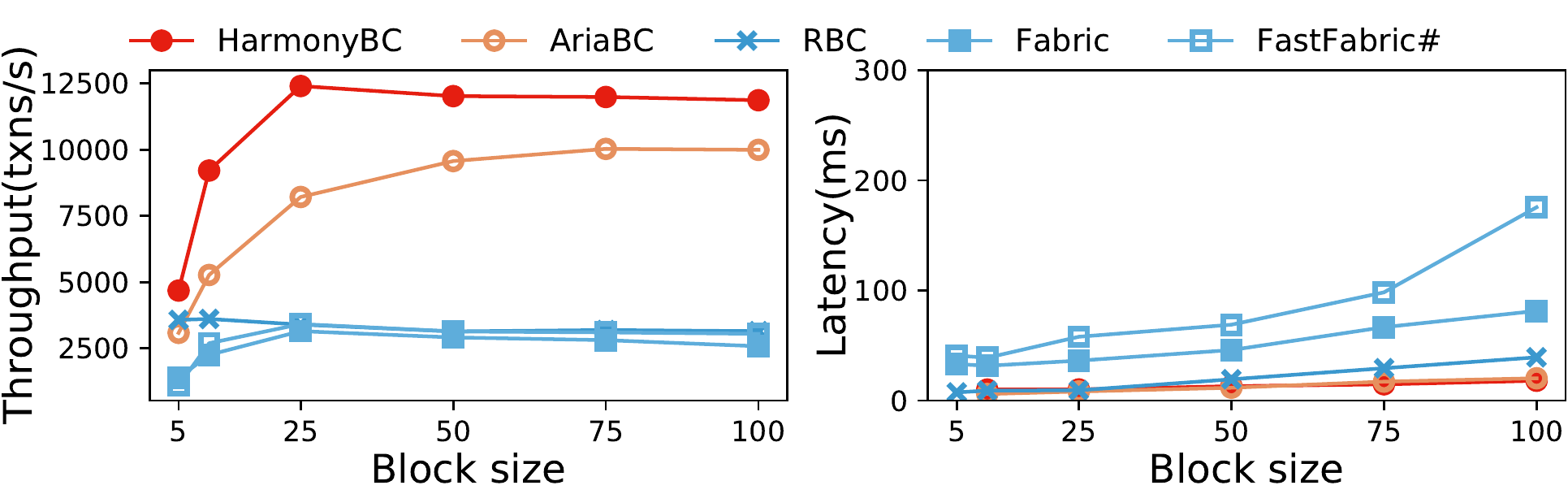}
     \vspace{-0.5em}
     \caption{Impact of block size on Smallbank} \vspace{-1em}
     \label{fig:smallbank_vary_block_size}
\end{figure}
\begin{figure}[t]
     \centering
     \includegraphics[width=0.9\linewidth]{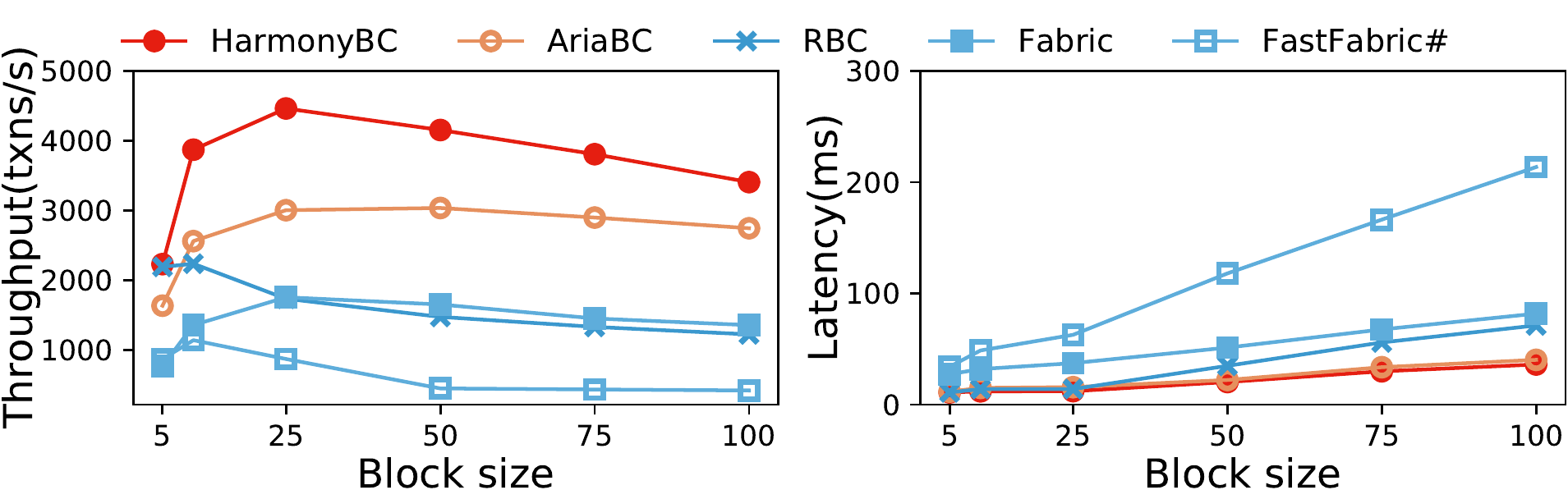}
     \vspace{-0.5em}
     \caption{Impact of block size on YCSB} \vspace{-1em}
     \label{fig:ycsb_vary_block_size}
\end{figure}
\begin{figure}[!t]
     \centering
     \includegraphics[width=0.9\linewidth]{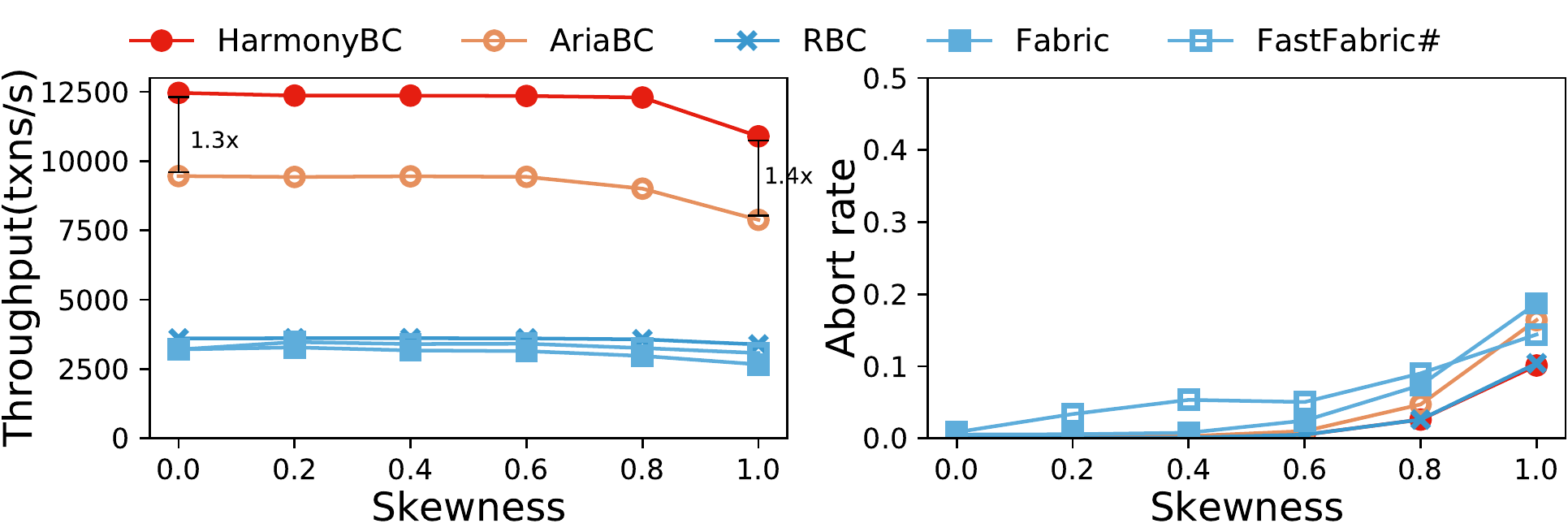}
     \vspace{-0.5em}
     \caption{Impact of contention on Smallbank}
     \vspace{-1em}
     \label{fig:smallbank_vary_skew}
\end{figure}
\begin{figure}[!t]
     \centering
     \includegraphics[width=0.9\linewidth]{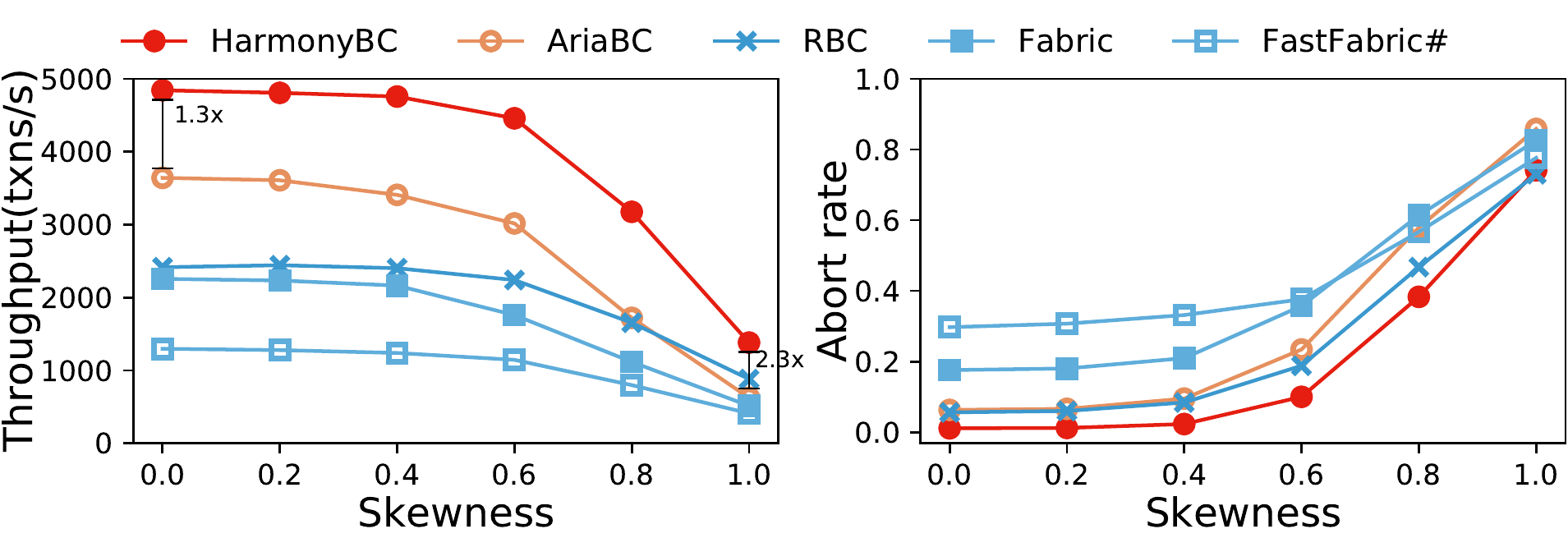}
     \vspace{-0.5em}
     \caption{Impact of contention on YCSB}
     \vspace{0.5em}
     \label{fig:ycsb_vary_skew}
\end{figure}

\subsection{The impact of contention and hotspot} \label{sec:exp_contention}
We first study the impact of contention by 
varying the skewness of the accessed data in our default cluster.
As shown in Figures \ref{fig:smallbank_vary_skew} and \ref{fig:ycsb_vary_skew}, 
all systems incur more aborts with a larger skewness and thus their performances drop (less severe in Smallbank since it generally has lower contention compared to YCSB).
Even when there is no skewness (i.e., skewness=0), Fabric and FastFabric\# abort many transactions in YCSB because the non-deterministic read-write sets often fail the clients to find a legitimate one for committing the transaction. FastFabric\# aborts more because in its implementation, it drops some transactions to avoid an overly large dependency graph.
Since AriaBC and RBC
abort a transaction on seeing a ww-dependency, and such dependencies occur fairly often even under low contention, HarmonyBC outperforms AriaBC and RBC under all skewnesses and it consistently has lower abort rates.
Both AriaBC and HarmonyBC have low abort rates in Smallbank.
But HarmonyBC outperforms AriaBC in throughput because it 
has fewer I/Os via update coalesce and 
better resource utilization via inter-block parallelism.

\added{Figure \ref{fig:false_abort} presents the false abort rates.
In general, the false abort rate grows with higher contention
except on YCSB with extreme contention ($\texttt{skewness=1.0}$) 
where most aborts in there are the real ones caused by serializability conflicts.
Despite that, we observe that 
Harmony has the lowest false rates in all cases.

}

Next, we particularly study the hotspot resiliency of HarmonyBC by using a variant of YCSB (we exclude Smallbank because it is not contentious enough).
There are still 10 statements in 1 transaction.
But we set 1\% of the records as hotspots, and 
a pair of \texttt{SELECT} and \texttt{UPDATE} operations 
would be rewritten as one \texttt{UPDATE} SQL statement that performs both read and write if they access the same record (since Postgres's optimizer does not have this rewrite rule, 
the hotspot resiliency feature of HarmonyDB has not been fully unleashed).
Each statement accesses a hotspot with a controlled probability. Fabric and FastFabric\# are not included in this experiment because they do not support SQL. As shown in Figure \ref{fig:hotspot}, HarmonyBC is almost unaffected by the hotspots, while the throughput of AriaBC and RBC drops significantly with increasingly higher abort rates when increasing the hotspot probability.

\begin{figure}[!t]
    \centering
    \begin{minipage}{\linewidth}
    \begin{subfigure}[b]{\linewidth}
        \centering
        \includegraphics[width=0.62\linewidth]{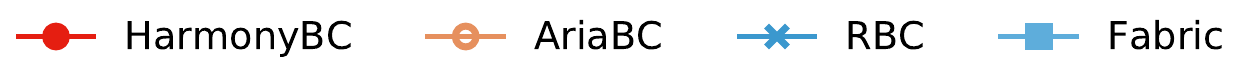}
    \end{subfigure}
    \end{minipage}
    \begin{minipage}{0.9\linewidth}
     \begin{subfigure}[b]{0.49\linewidth}
        \centering
        \includegraphics[width=\linewidth]{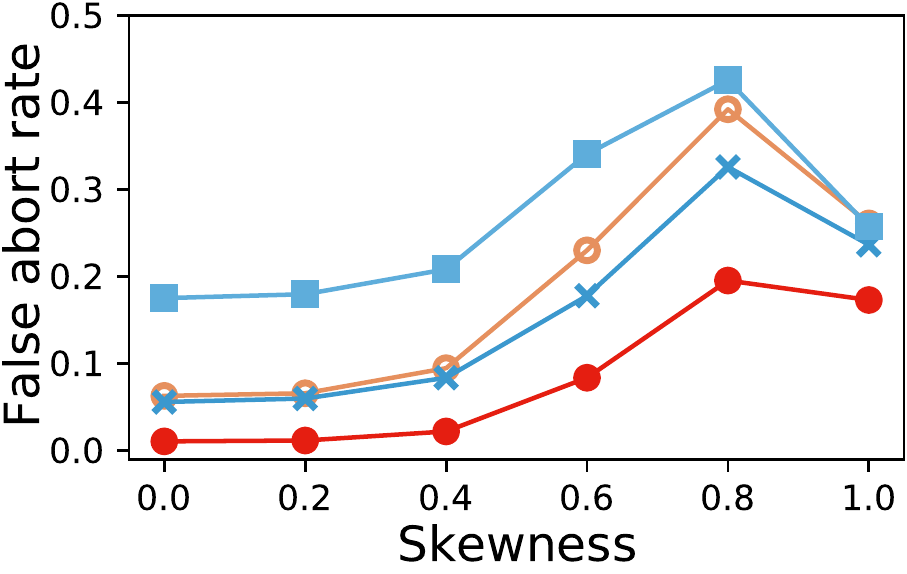}
        \caption{YCSB}\vspace{-0.15cm}
        \label{fig:ycsb_false_abort}
    \end{subfigure}
    \hfill
    \begin{subfigure}[b]{0.49\linewidth}
        \centering
        \includegraphics[width=\linewidth]{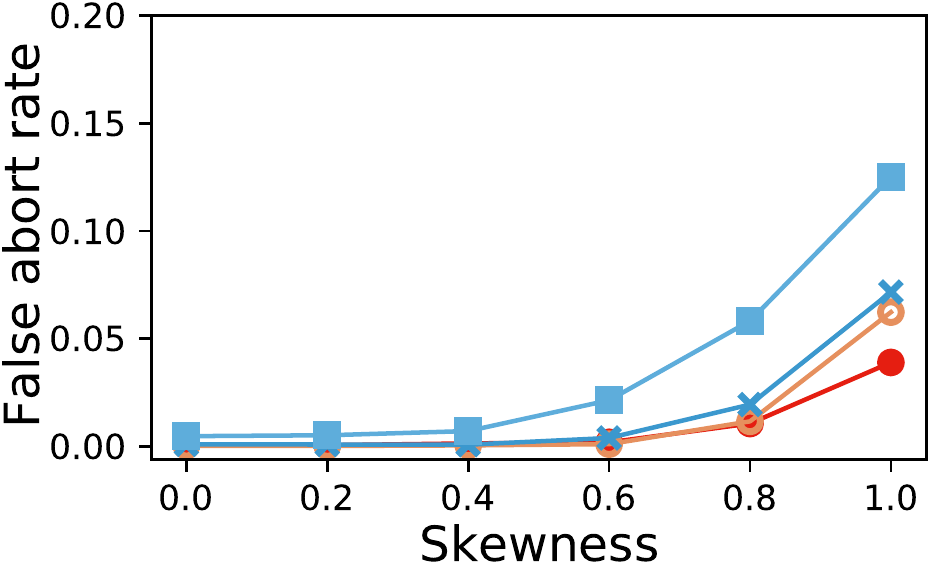}
        \caption{Smallbank}\vspace{-0.15cm}
        \label{fig:smallbank_false_abort}
    \end{subfigure} 
    \end{minipage}
    \caption{\added{False abort rate.  FastFabric\# is excluded because it uses an expensive unparallelizable graph traversal to eliminate all false aborts.}}
    \vspace{-0.3cm}
    \label{fig:false_abort}
\end{figure}

\begin{figure}[!t]
     \centering
     \includegraphics[width=0.9\linewidth]{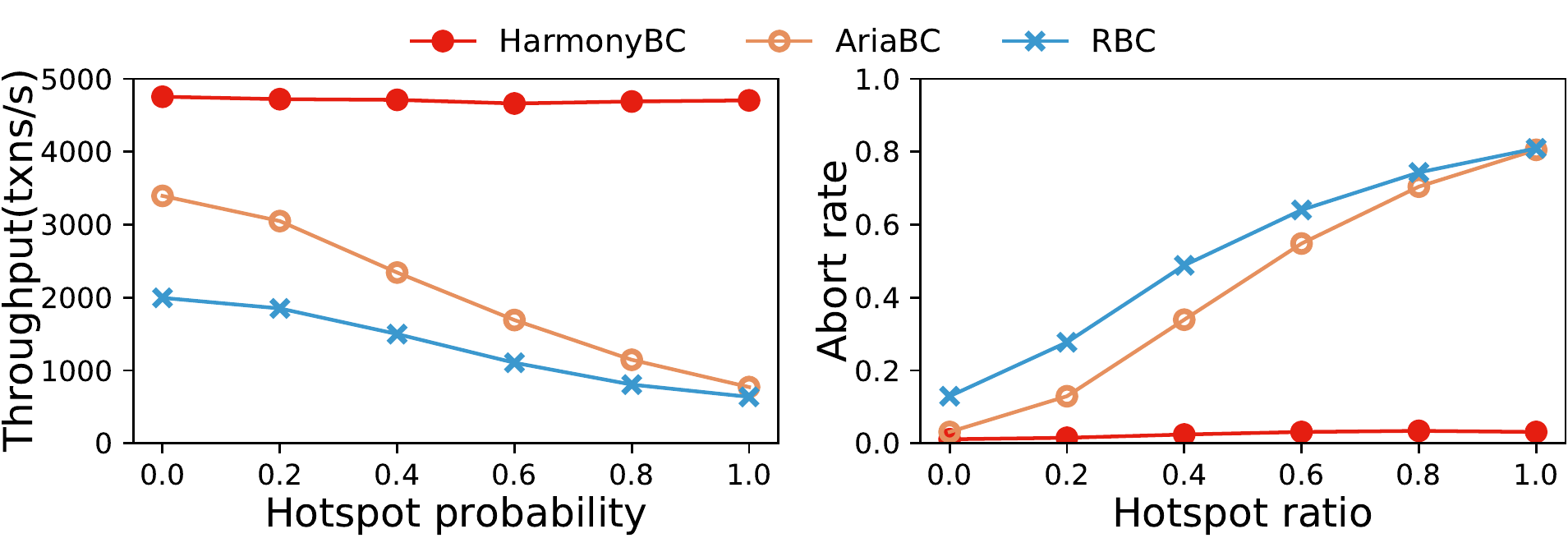}
     \vspace{-0.1cm}
     \caption{Impact of hotspots}\vspace{0.5em}
     \label{fig:hotspot}
\end{figure}

\subsection{The impact of number of replicas}\label{sec:exp_replica}

In this experiment, we scale the number of replicas up to 80
using the cloud cluster with 3 of them as the ordering service. All instances in this experiment
are located in the same region.
Figures \ref{fig:smallbank_vary_peers} and \ref{fig:ycsb_vary_peers} show the results.
The throughput and latency of HarmonyBC, AriaBC and RBC are almost unaffected by the number of replicas because their replicas work independently and only receive small transaction commands from the ordering service.
\deleted{
Although using the OE architecture, RBC's throughput decreases with more replicas.
That is because in RBC, the clients need to request the latest block number from one of the replicas when issuing a transaction. The latest block number has a lower chance to keep in-sync among the replicas due to network asynchrony when there are more replicas, causing more stale reads if the block number is actually lagged.
}
All SOV blockchains, in contrast, exhibit throughput drop and latency increase with more replicas due to a larger network overhead of transferring read-write sets.

\subsection{Geo-distributed cluster with BFT consensus} \label{sec:bft} 
In this experiment, we study the performance of HarmonyBC using Byzantine-fault-tolerant (BFT) HotStuff \cite{hotstuff} as the consensus layer 
and show that BFT would not hurt its throughput even scale up to 80 geo-distributed consensus nodes.

\begin{figure}[t]
     \centering
     \includegraphics[width=0.9\linewidth]{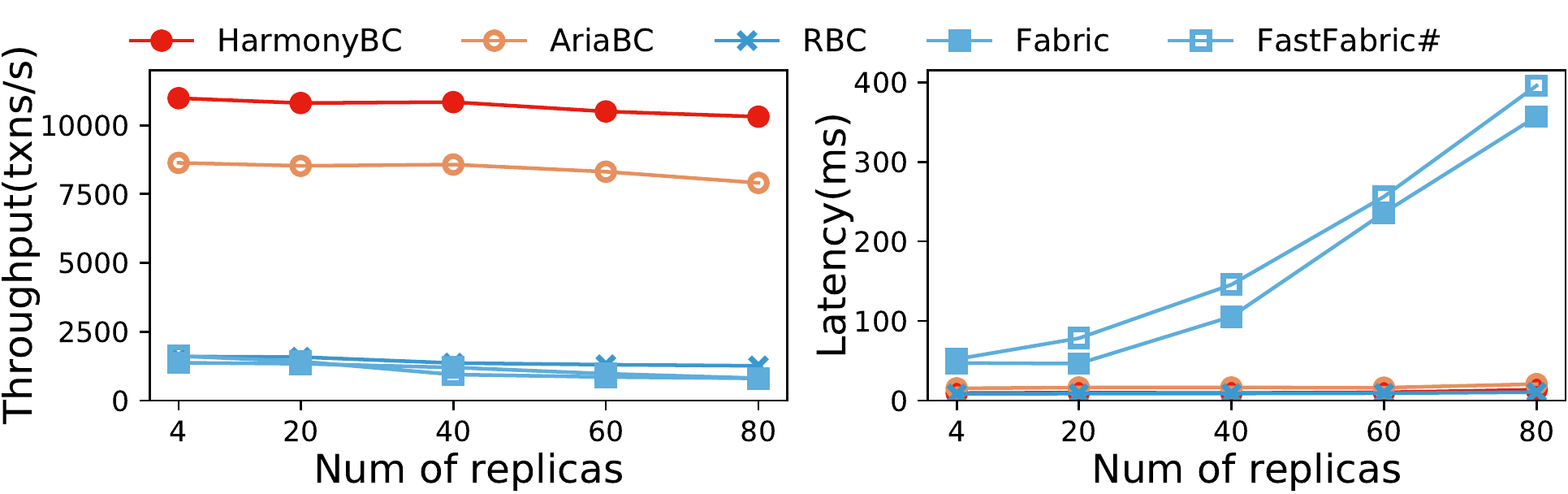}
     \vspace{-0.5em}
     \caption{Impact of number of replicas on Smallbank}
     \vspace{-1em}
     \label{fig:smallbank_vary_peers}
\end{figure}

\begin{figure}[t]
     \centering
     \includegraphics[width=0.9\linewidth]{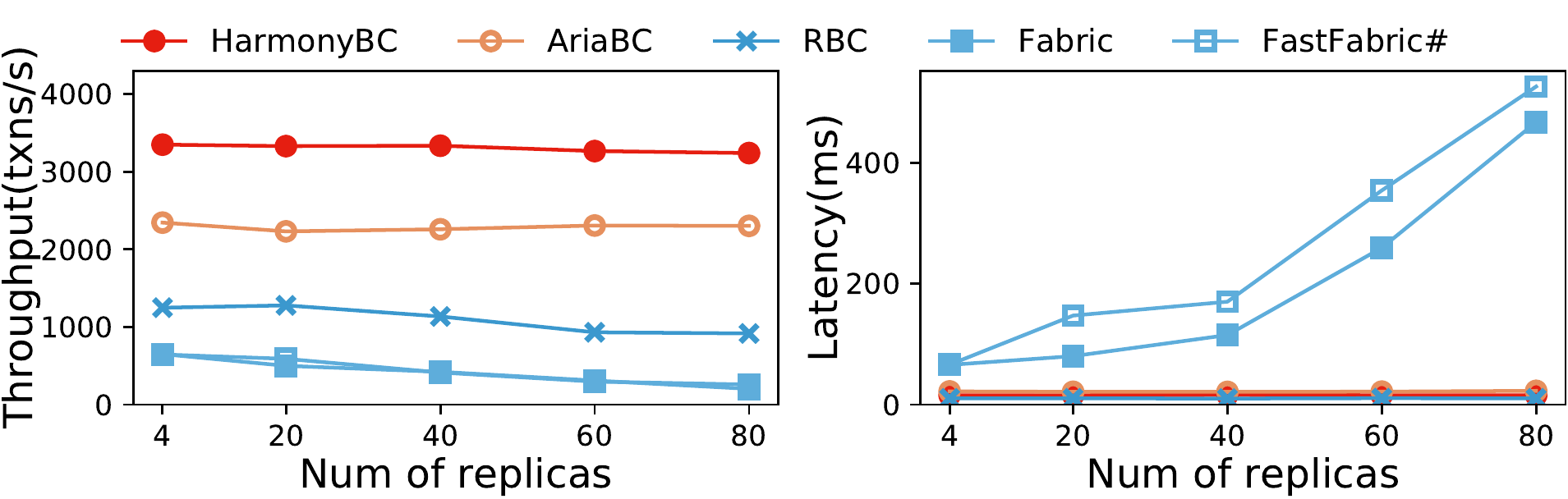}
     \vspace{-0.5em}
     \caption{Impact of number of replicas on YCSB} 
     \vspace{-1em}
     \label{fig:ycsb_vary_peers}
\end{figure}

\begin{figure}[t]
     \centering
     \includegraphics[width=0.92\linewidth]{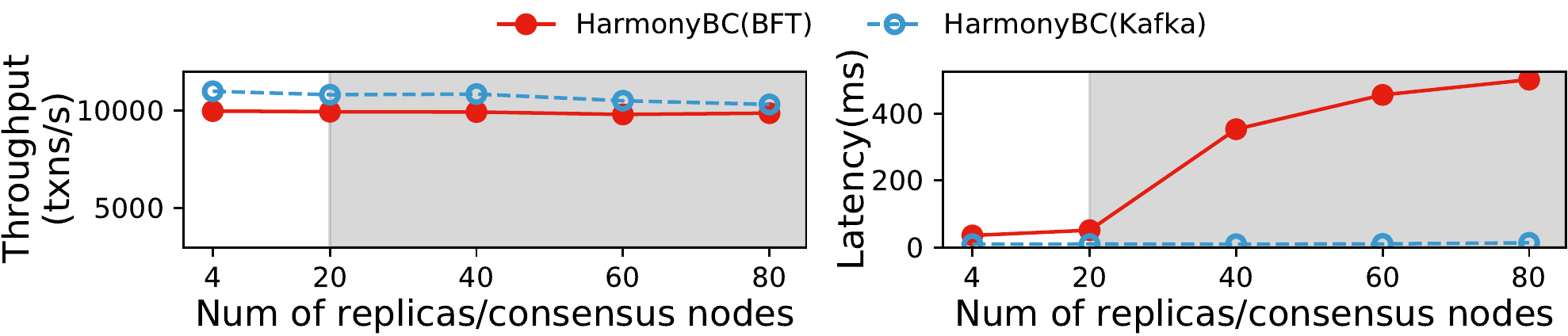}
     \vspace{-0.5em}
     \caption{Impact of BFT consensus on smallbank}
     \vspace{-1em}
     \label{fig:bft-smallbank}
\end{figure}

\begin{figure}[t]
     \centering
     \includegraphics[width=0.92\linewidth]{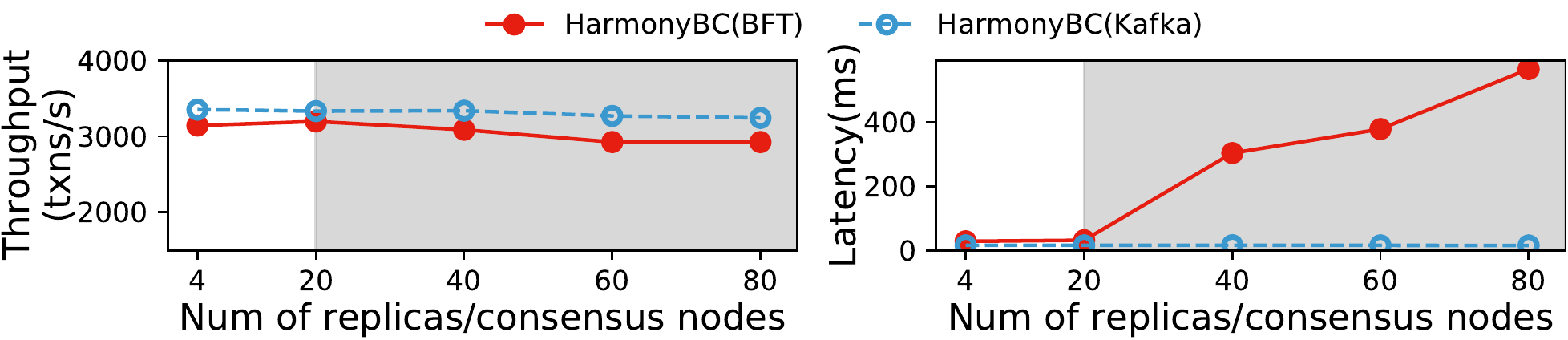}
     \vspace{-0.5em}
     \caption{Impact of BFT consensus on YCSB}
     \vspace{0.5em}
     \label{fig:bft-ycsb}
\end{figure}

To do so, we use the cloud cluster and allocate 80 instances from 4 continents
(Ohio, Mumbai, Sydney, and Stockholm; 20 instances each). For each instance, we run a consensus node and a replica node.
Figures \ref{fig:bft-smallbank} and Figure \ref{fig:bft-ycsb} show the results.
For comparison, the performance of the non-BFT HarmonyBC obtained in Section \ref{sec:exp_replica} is also shown in the figures.
In the experiment, the first 20 instances are located in the same region (the white part), and experiments with more than 20 instances (the gray part) show the results of the geo-distributed setting.
The results show that the throughput of HarmonyBC using a BFT consensus is almost unaffected (the minor throughput drop is due to HotStuff's other CPU overhead such as crypto). 
Since HotStuff requires more network round-trips to establish consensus than Kafka, the latency becomes larger and increases with more nodes. However, the increased latency does not affect the throughput because the consensus layer is not the bottleneck in this setting.

\subsection{TPC-C results}\label{sec:tpcc}
\added{
In this experiment, we study the performance of Harmony using TPC-C in our default cluster. Fabric and FastFabric\# are excluded here because they do not support the relational model natively \cite{fabric-tpcc}.
Figure \ref{fig:tpcc} shows the results in TPC-C  
are consistent with the ones in YCSB and Smallbank, 
where HarmonyBC outperforms the others under all settings.
HarmonyBC outperforms the others the most (3.3$\times$) 
when the workload contention is the highest (1 warehouse), 
which reflects HarmonyDB's resiliency to hotspot and the effectiveness of its abort-minimizing design.
Although with less contention, increasing the number of warehouses would increase the database size.
That explains the throughput and latency trend 
for more than 20 warehouses.
}

\subsection{Ablation study} \label{sec:effectiveness}
\begin{figure}
     \centering
     \includegraphics[width=0.9\linewidth]{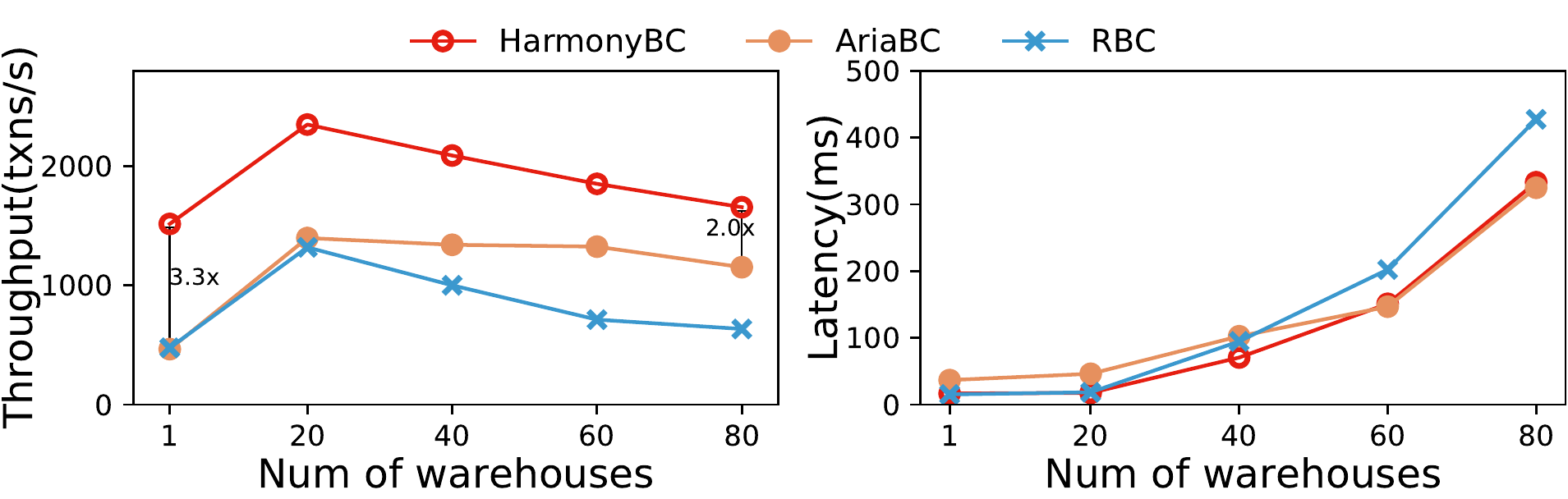}
     \vspace{-0.5em}
     \caption{\added{TPC-C results}}
     \vspace{-1em}
     \label{fig:tpcc}
\end{figure}
\begin{figure}
     \centering
     \includegraphics[width=0.95\linewidth]{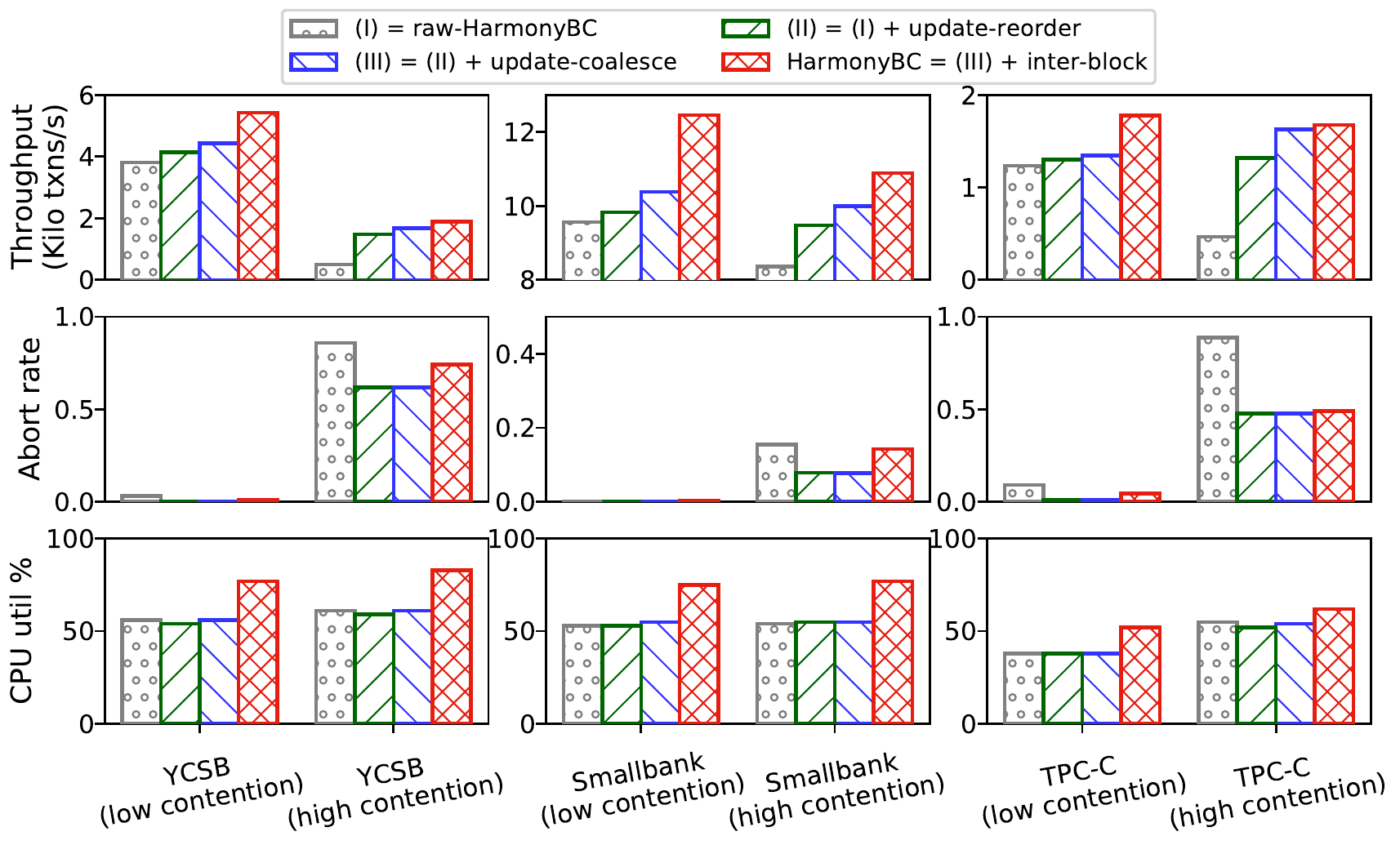}
     \vspace{-0.5em}
     \caption{\added{Ablation Study under Low contention (0 skewness in YCSB/Smallbank; 80 warehouses in TPC-C) and High contention (1.0 skewness in YCSB/Smallbank; 1 warehouses in TPC-C)}}
     \vspace{1em}
     \label{fig:effectiveness}
\end{figure}
\deleted{
In this experiment, we turn on and off some HarmonyBC's designs to demonstrate their effectiveness in our default cluster.
Since update reordering and update coalescence work in unison as shown in Algorithm \ref{alg:apply-writes}, we do not further break them down. 
When \texttt{update-optim} (including update reordering and coalescence) is disabled,  Harmony resolves conflicting updates by aborting a transaction on seeing a ww-dependency like Aria.
In Figure \ref{fig:effectiveness}, we first disable both \texttt{update-optim} and \texttt{inter-block} (i.e. inter-block parallelism) and add them back progressively.

In Smallbank, inter-block parallelism contributes to the major part of the improvement because of better CPU utilization as we measured an increase from 52\% to 68\% when skewness = 0 
(vanilla Postgres obtains 82\% utilization on the same workload; HarmonyBC cannot reach that utilization even with inter-block parallelism due to the barrier between the simulation step and the commit step required for determinism).
That confirms that it is worth enabling inter-block parallelism
because the better resource utilization outweighs 
the minor increase of aborts introduced by 
inter-block dependencies, even under high contention (skewness = 1.0).
\texttt{update-optim} does not help much in Smallbank because
the number of conflicting updates is too small for \texttt{update-optim} to shine.
\texttt{update-optim} becomes the major performance booster in YCSB under high contention (e.g., $3.2\times$ in YCSB when skewness is 1.0) since YCSB has more update operations per transaction and a high skewness causes more conflicting updates.
}
\added{
In this experiment, we conduct an ablation study to demonstrate the effectiveness of each individual optimization technique.
We cannot disable abort-minimizing validation because it is essential for correctness.
When we disable update reordering, 
we use back Aria's style to abort a transaction on seeing a ww-dependency to preserve correctness.

Figure \ref{fig:effectiveness} shows the results,
where {\tt raw-HarmonyBC} means all optimizations of HarmonyBC are disabled (except abort-minimizing validation).
From the top row, we see that each individual optimization can offer throughput improvement on all workloads,
although in different degrees.
Inter-block parallelism (red bar) and update-reordering (green bar) 
are shown to be the most effective optimizations under low contention and high contention, respectively.
Inter-block parallelism improves low contention throughput by increasing CPU utilization (3rd row) 
whereas update-reordering improves high contention throughput by reducing the abort rate (2nd row).
Notice that the use of inter-block parallelism 
increases the abort rate a bit,
but it increases the overall throughput as a whole. 
It shows that better CPU utilization weighs higher than the other factors in this case.
}

\subsection{Is Harmony still useful if all disk-related overheads are gone?}\label{sec:disk_overhead}
\added{
HarmonyDB is a disk-oriented blockchain.
A natural question is whether Harmony's optimizations are still effective if all disk-related overheads in the database layer are gone.  
As pointed out by Stonebraker et al. \cite{oltp}, 
the cost spent on a disk-oriented database includes 
(i) the disk I/O latency and 
(ii) the cost of masking the I/O latency (e.g., buffer manager, 2PL).
Hence, to answer that question,
we carried out an experiment 
that replaces SSD with a RAMDisk to remove (i).
PGSQL, the disk database that we build HarmonyBC on top of, has all sorts of costs on (ii) to minimize (i).
Hence, we also implement Harmony's techniques as a standalone memory database engine to study the effect of Harmony without both (i) and (ii).

Figure \ref{fig:ramdisk} shows the results on all three workloads.
For each workload, 
the two bars on the left are the results from the setting that we have been targeting -- SSD-oriented blockchains.
The two bars in the middle are the results of running the same
but replacing the SSDs with RAMDisk.
We can see that the throughput of HarmonyBC (as well as AriaBC) increases when there is no disk I/O.
We can also see that all Harmony's optimizations are still effective there.
With update reordering, fewer transactions are aborted; with inter-block parallelism, transactions from a later block 
can run with transactions from the previous block to increase CPU utilization. Both lead to higher throughput.
Of course, as pointed out in \cite{oltp},
many components become unnecessary (e.g. buffer manager) or needed to be redesigned if using memory as the main storage.
Those components exactly are the reasons why 
HarmonyBC's throughput cannot exceed the throughput of native memory blockchains (the two bars on the right).
Yet, even though native memory blockchains can have the best throughput, 
their throughput is really capped by the consensus layer
and it is still unclear whether the industry is willing to spend the DRAM money on that.}

\begin{figure}
     \centering
     \includegraphics[width=\linewidth]{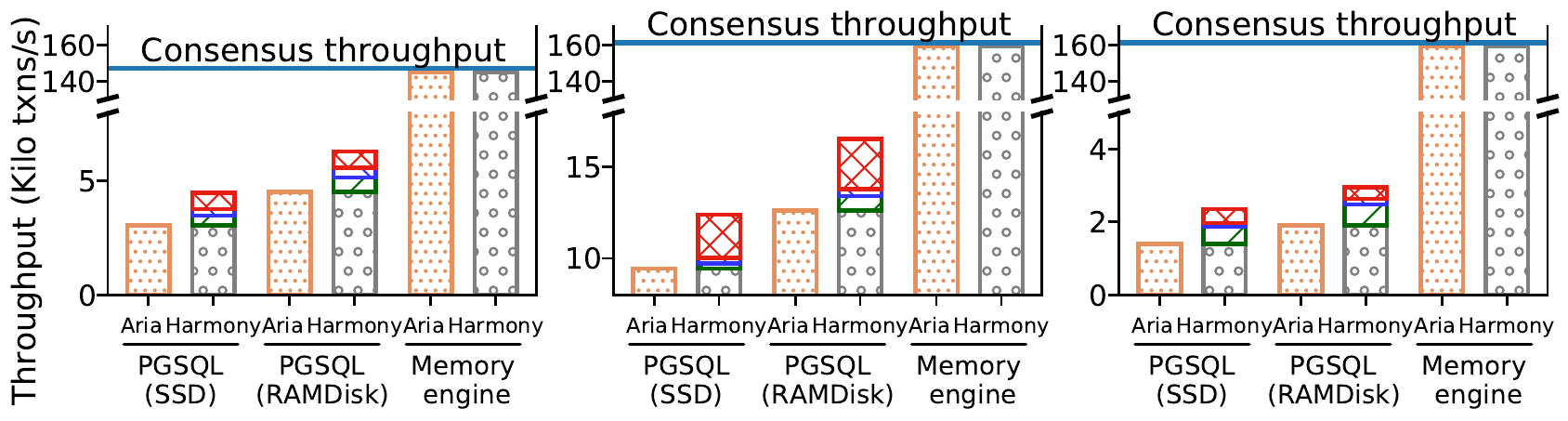}
     \vspace{-1em}
     \begin{minipage}{\linewidth}
     \begin{subfigure}[t]{0.33\linewidth}
     \caption{YCSB}
     \end{subfigure}
     \hfill
     \begin{subfigure}[t]{0.33\linewidth}
     \caption{Smallbank}
     \hfill
     \end{subfigure}
     \begin{subfigure}[t]{0.3\linewidth}
     \caption{TPC-C}
     \end{subfigure}
     \end{minipage}
     \vspace{-0.5em}
     \caption{\added{Blockchain throughput difference between non-disk and disk-based database layer.  
     The breakdown of Harmony's throughput is also given, following the format used in the ablation study.
     The consensus throughput varies across workloads due to the different transaction sizes.}\vspace{0.5em}}
     \label{fig:ramdisk}
\end{figure}
\section{Related Work} \label{sec:related}
\stitle{Blockchain consensus.} 
Efforts have been devoted to improving the consensus layer of blockchains [\added{\citenum{cop}}, \citenum{red-belly, mir-bft, proof-of-execution, fire-ledger, zyzzyva, bocwbo, HQ, ringbft}].
Most of them
use an in-memory hash map to simulate the state machine \cite{red-belly, mir-bft, proof-of-execution, ringbft, zyzzyva}.
Main-memory blockchains \cite{RCC, resilientdb} use DRAM as the storage layer and rely on replicas for fault tolerance --- a setting orthogonal to Harmony.
Without the disk-based database engine as the bottleneck,
they can attain a throughput as high as 100K tps.
Yet, disk-based blockchains still own the lion's
share of the market given their economies of scale
HarmonyBC aims to narrow the performance gap between the database layer and the consensus layer.

\stitle{Private Blockchain.}
Serial transaction execution cannot utilize multi-core CPUs for higher blockchain throughput (e.g., Quorum \cite{quorum}, Diem \cite{diem}, Concord \cite{concord}).
With a high-performance BFT consensus layer,
efficient deterministic concurrency becomes crucial to achieve better end-to-end throughput on disk-oriented private blockchains.
Besides those covered in Section \ref{sec:connect}, 
\added{Eve \cite{all-about-eve} is based on non-deterministic concurrency control. To upkeep replica consistency, it carries out consensus after concurrent execution. 
Since divergences happen frequently under non-deterministic execution,
Eve also uses static analysis to avoid scheduling conflicting transactions to the same block, which inherits the application limitations from static analysis.}
SChain \cite{schain} is a private blockchain with inter-block parallelism, but it requires static analysis and 
hence with limited applications.
Some private blockchains use sharding 
 to scale out \cite{sharding,sharding-2,sharding-3,byshard,caper,blockchaindb,basil}.
Sharding is part of our future work.
Trusted hardware (e.g., Intel SGX) can also be utilized to improve blockchain performance \cite{sharding-3,fabric-tee,tee-bft,chain-of-trust,tee-2} and privacy \cite{ccf,corda,tee-1}. 
The use of trusted hardware is also orthogonal to Harmony.

\stitle{Deterministic database.} 
Other works in deterministic databases focus on contention management \cite{caracal} and more efficient deterministic scheduling \cite{quecc,single-thread-dcc,tpart}. 
Some deterministic databases are able to eliminate two-phase commit (2PC) for cross-shard transactions \cite{calvin,bohm,pwv}. However, they all require static analysis. Aria \cite{aria} does not require static analysis but then a light form of 2PC is required.

\stitle{Distributed database.}
The connection between private blockchains and distributed databases has been studied both in terms of design \cite{bc_vs_db_1, bc_vs_db_2, bc_vs_db_3, bc_vs_db_4, bc_vs_db, untangling} and performance \cite{blockbench}.
Fusing the optimization techniques from distributed database into private blockchains is an ongoing trend \cite{fabric++, fabricsharp, fastfabric, rbc, chainifydb}. However, the close relationship between private blockchain and deterministic database has never been studied explicitly. We are the first to fill this gap and develop HarmonyBC to adopt deterministic concurrency control in private blockchains.

\stitle{Database concurrency.}
Deterministic concurrency control is built on the shoulder of many modern concurrency control protocols \cite{batch-occ,atomic-piece,ermia,cicada,silo,mocc,tictoc,opportunities,bamboo,bcc}.
Harmony shares thoughts with serializable snapshot isolation (SSI) \cite{ssi,ssi-2,ssi-3} in terms of simulating transactions against snapshots.
However,
SSI has a first-committer-wins rule that aborts a transaction if two transactions update the same record.
In contrast, Harmony reorders updates on the same record without any abort.
Harmony is also related to batch-based OCCs \cite{silo, batch-occ} and group commit \cite{group-commit-1,group-commit-2}. 
They improve efficiency by batching some operations (e.g., logging) while transactions stream in. In contrast,
Harmony's input is already batched due to the blockchain setting. 
STOv2 \cite{opportunities} proposes commit-time update to reduce conflicts in OCCs. In \cite{lazy}, updates are even deferred until the updated records are read by the others. Harmony also applies updates during commit, but it is superior by further reducing conflicts using update reordering to re-organize the dependencies.

\section{Conclusions} \label{sec:conclusion}
Many blockchains are still disk-oriented because of cost and use cases.
In this paper, we connect private blockchain with deterministic database and propose Harmony, a deterministic concurrency control protocol that is inspired by the state-of-the-art deterministic databases but specifically optimized for 
blockchain's disk-oriented database layer. 
HarmonyBC, our private blockchain prototype achieves 2.0$\times$ to 3.5$\times$ better throughput compared with the state-of-the-art private blockchains. 

\begin{acks}
This work is partially supported by Hong Kong General Research Fund (14200817), Hong Kong AoE/P-404/18, Innovation and Technology Fund (ITS/310/18, ITP/047/19LP) and Centre for Perceptual and Interactive Intelligence (CPII) Limited under the Innovation and Technology Fund.
\end{acks}


\end{document}